\RequirePackage{fix-cm}
\documentclass[smallcondensed,numbook,envcountsame]{svjour3}     

\usepackage{graphicx}

\usepackage{latexsym,bm,amsmath, amssymb,amsfonts}%
\usepackage{amsmath,epsfig}
\usepackage{cases} 
\usepackage{amssymb,cite}
\usepackage{epsfig}
\usepackage{color}
\usepackage{bm}
\usepackage{cite}
\usepackage{multirow}
\usepackage[format=plain]{caption}
\usepackage[utf8]{inputenc}
\usepackage{enumitem}
\usepackage{algpseudocode}
\usepackage{algorithm}
\usepackage{booktabs}

\newtheorem{thm}{Theorem}
\newtheorem{lem}{Lemma}
\newtheorem{prop}{Proposition}
\newtheorem{coro}{Corollary}

\begin{document}

\title{
\textcolor{black}{Optimal Estimation of Sensor Biases for Asynchronous Multi-Sensor Data Fusion}
}


\author{Wenqiang Pu \and Ya-Feng Liu \and Junkun Yan \and Hongwei Liu \and Zhi-Quan Luo}


\institute{This work is performed at the Shenzhen Research Institute of Big Data, The Chinese University of Hong Kong, Shenzhen. {This work is funded in part by a National Natural Science Foundation of China (NSFC) Key Project grant 61731018, in part by NSFC grants 11331012, 11631013, and 61601340, and in part by the China National Funds for Distinguished Young Scientists grant 61525105.}\\[5pt]
W. Pu, J. Yan, and H. Liu \at National Laboratory of Radar Signal Processing, \\
	Collaborative Innovation Center of Information Sensing and Understanding, \\
	Xidian University, Xi’an 710071, China \\
	\email{wqpu@stu.xidian.edu.cn; jkyan@xidian.edu.cn;
	hwliu@xidian.edu.cn}
           \and
           Y.-F.~Liu \at State Key Laboratory
of Scientific and Engineering Computing,\\
Institute of Computational
Mathematics and Scientific/Engineering Computing,\\
 Academy of Mathematics and Systems Science, \\
           Chinese Academy of Sciences, Beijing 100190, China \\
           \email{yafliu@lsec.cc.ac.cn}
              \and
           Z.-Q.~Luo \at Shenzhen Research Institute of Big Data, \\
           The Chinese University of Hong Kong, Shenzhen 518100, China \\
           E-mail: luozq@cuhk.edu.cn
}

\date{Received: date / Accepted: date}

\maketitle

\begin{abstract}
An important step in a multi-sensor surveillance system is to estimate sensor biases from their noisy asynchronous measurements. This estimation problem is computationally challenging due to {the highly} nonlinear transformation between the global and local coordinate systems as well as the measurement \textcolor{black}{asynchrony} from different sensors. In this paper, we propose a novel nonlinear {least squares} (LS) formulation for the problem by assuming the existence of a \textcolor{black}{reference target moving with an (unknown) constant velocity}. We also propose \textcolor{black}{an efficient} block coordinate decent (BCD) optimization algorithm, with a judicious initialization, to solve the problem. \textcolor{black}{The proposed BCD algorithm alternately updates the range and azimuth bias estimates by solving linear {least squares} problems and semidefinite programs {(SDPs)}.} In the absence of measurement noise, the proposed algorithm is guaranteed to find the global solution of the problem and the true biases. Simulation results show that the proposed algorithm significantly outperforms the existing approaches in terms of the root mean square error (RMSE).


\end{abstract}


\section{Introduction}
\label{sec:intro}
\subsection{Background and Motivation}
\textcolor{black}{With the proliferation of low cost sensors, there is an increasing interest in integrating stand-alone sensors into a multi-sensor system for various engineering applications such as surveillance \cite{bar1990multitarget}}.
Instead of developing expensive high-performance sensors, directly fusing data from existing multiple inexpensive sensors is a more cost-effective approach to improving the performance of tracking and surveillance systems.
An important process in multi-sensor integration is \emph{registration} (or \emph{alignment}) \cite{Dana1990Registration}, \textcolor{black}{whereby the multi-sensor data is expressed in a common reference frame, by removing the sensor biases caused by antenna orientation and improper alignment \cite{fischer1980registration}.}
Since sensor biases change slowly with time \cite{fischer1980registration}, they can be treated as constants during a relatively {long} period of time. Consequently, the sensor registration problem is usually an estimation problem for sensor biases. \textcolor{black}{In general, the estimation problem is computationally challenging due to the   nonlinear transformation between the global and the local coordinate systems of the sensors as well as measurement asynchrony from different sensors.}
\subsection{Related Works}
\textcolor{black}{Both synchronous \cite{fischer1980registration,helmick1993removal,zhou1997an,Ristic2003Sensor,Fortunati2011Least} and asynchronous \cite{zhou2004asy,lin2005multisensor,Lin2006Multisensor,ristic2013calibration} sensor registration problems have been studied in the literature.
	In the synchronous case, all sensors simultaneously observe the target at the same time instances, whereas in the asynchronous case,  the sensors observe the target at different time instances.}

\textcolor{black}{For the synchronous registration problem,} \cite{fischer1980registration} identified various factors which dominate the alignment errors in the multi-sensor system and established a sensor bias model. \textcolor{black}{Under the assumption that there exists a bias-free sensor,} a maximum likelihood (ML) estimation formulation \cite{fischer1980registration,Fortunati2013On} and a nonlinear {least squares} (LS) estimation formulation \cite{Fortunati2011Least} were presented for this problem, and various algorithms were proposed to approximately solve the problem including the expectation maximization (EM) algorithm \cite{Fortunati2013On} and the successive linearization ({least squares}) algorithm \cite{Fortunati2011Least}. \textcolor{black}{In the absence of a bias-free sensor,} \textcolor{black}{several approaches  \cite{leung1994least,Dana1990Registration,Cowley1993Registration,zhou1997an,Ristic2003Sensor,Hsieh1999Optimal,Nabaa1999Solution,Zia2008EM,Li2010Joint,Okello2004Joint} have also} been proposed. These approaches, from the perspective of parameter estimation \textcolor{black}{point}, can be divided into the following three types: LS, ML and Bayesian estimation. Both of the LS and ML approaches treat biases as deterministic variables while the Bayesian approach treats biases as random variables. {Note that there are two different kinds of variables, target positions and sensor biases, in the sensor registration problem.} The difference between the LS and ML approaches lies in the way of dealing with target positions.
In the LS approaches \cite{leung1994least,Dana1990Registration,Cowley1993Registration}, target positions are approximately represented and eliminated in the LS formulation and finally only sensor biases are estimated.
\textcolor{black}{In contrast, the ML approaches jointly estimate target positions and sensor biases \cite{zhou1997an,Ristic2003Sensor}. In particular, an iterative two-step optimization algorithm with some approximations was proposed in \cite{zhou1997an} to solve the ML formulation, with one step for estimating target positions and the other for estimating sensor biases.}
The Bayesian approach treats sensor biases as random variables and integrates target tracking and biases estimation in a Bayesian estimation framework \cite{Hsieh1999Optimal,Li2010Joint,Nabaa1999Solution,Zia2008EM,Okello2004Joint}. Different strategies were developed to decouple the two estimation tasks and to deal with the nonlinearity, including  two-stage Kalman filter (KF) \cite{Hsieh1999Optimal}, EM-KF \cite{Li2010Joint}, and EM particle filter \cite{Zia2008EM}.

\textcolor{black}{In practice, sensors are often not synchronized in time due to different data rates. The asynchronous sensor measurements make the estimation problem underdetermined.
	To overcome this difficulty, researchers have exploited \textit{a priori} knowledge of the target motion model such as the \textcolor{black}{\textit{nearly-constant-velocity} motion model \cite{li2003survey}.}
	Based on this side information, recursive Bayesian approaches for jointly estimating target states and sensor biases have been proposed in  \cite{zhou2004asy,lin2005multisensor,Lin2006Multisensor,Taghavi2016practical,mahler2011bayesian,ristic2013calibration}. \textcolor{black}{In particular, estimating sensor biases from asynchronous measurements based on an approximated linear measurement model was studied in \cite{zhou2004asy,lin2005multisensor,Lin2006Multisensor,Taghavi2016practical}.} However, the approximation procedure in these approaches requires the sensor biases to be small.} In references \cite{mahler2011bayesian,ristic2013calibration}, a unified Bayesian framework for target tracking and biases estimation was proposed based on probability hypothesis density (PHD) \cite{Mahler2003Multitarget}. This approach relies on the implementation of particle PHD filter which can result in high computational cost when the number of sensors is large.

\textcolor{black}{The key difficulty of solving the sensor registration problem is the intrinsic nonlinearity of the problem. Although various approaches, such as those based on different linear approximation \cite{fischer1980registration,helmick1993removal,zhou1997an,lin2005multisensor,zhou2004asy,Ristic2003Sensor,Fortunati2011Least,Fortunati2013On,Dana1990Registration,leung1994least,Lin2006Multisensor,Taghavi2016practical} and the particle filtering \cite{Zia2008EM,mahler2011bayesian,ristic2013calibration}, have been proposed to deal with the nonlinearity issue, 
they either suffer from a model mismatch error, which might significantly degrade the estimation performance, or incur a prohibitively high computational cost.}


\subsection{Our Contributions}
In this paper, \textcolor{black}{we consider the sensor registration problem and focus on a general scenario where all sensors are possibly biased
and operate asynchronously. Among all possible sensor biases \cite{fischer1980registration}, we consider two major and fundamental ones, i.e., range and azimuth biases. Instead of linear approximations, we deal with the intrinsic nonlinearity issue in the sensor registration problem by using nonlinear optimization techniques.} The  main contributions of the paper are as follows:
%
%
%
\begin{enumerate}[fullwidth,label=(\arabic*)]
	\item \textit{\textcolor{black}{A New Nonlinear LS Formulation:}} We propose a new nonlinear LS formulation for the asynchronous multi-sensor registration problem by \textcolor{black}{only} assuming the \textit{a priori} information that the target moves \textcolor{black}{with a {nearly-constant-velocity}  model\cite{li2003survey}. Our proposed formulation takes different practical uncertainties into consideration, i.e., noise in sensor measurements and dynamic maneuverability in the target motion.}
%
	\item \textit{\textcolor{black}{Separation Property of Range Bias Estimation:}} \textcolor{black}{We reveal an interesting separation property of the range bias estimation.} More specifically, we show that the solution of the problem of minimizing the LS error over the range bias for each senor can be decoupled from its azimuth bias. This property sheds an important insight that each sensor can estimate its range bias from its local measurements.
	\item \textit{An Efficient BCD Algorithm and Exact Recovery:} We fully exploit the special structure of the proposed nonlinear LS formulation and develop an efficient {block coordinate decent (BCD)} algorithm, with a judicious initialization by using the separation  property of range bias estimation. The proposed algorithm alternately updates the range and azimuth biases by solving linear \textcolor{black}{LS} problems and semidefinite programs (SDPs). We show, in Theorem \ref{thm:multi}, that solving the original non-convex problem with respect to the azimuth biases is equivalent to solving an SDP, which itself is interesting. We also show exact recovery of the proposed BCD algorithm in the noiseless case, i.e., our algorithm is able to find the true (range and azimuth) biases if there is no noise \textcolor{black}{in the target motion and sensor measurements}.
\end{enumerate}

We adopt the following \textcolor{black}{notation} in this paper. Lower and upper case letters in bold are used for vectors and matrices, respectively. For a given matrix $\mathbf{X},$ we denote its transpose and Hermitian transpose by $\mathbf{X}^T$ and $\mathbf{X}^\dagger$, respectively. If $\mathbf{X}$ is further a square matrix, we use $\textrm{diag}(\mathbf{X})$ to denote the column vector formed by its diagonal entries, use $\textrm{Tr}(\mathbf{X})$ to denote its trace, use $\lambda_{{\min}}(\mathbf{X})$ to denote its smallest eigenvalue, {{use $\text{vec}(\mathbf{X})$ to denote the vector sequentially stacked by its columns, use $\mathbf{X}\succ \mathbf{0}$ ($\mathbf{X}\succeq \mathbf{0}$) to denote that it is a positive definite (semi-definite) matrix}}, and use $\mathbf{X}^{-1}$ to denote its inverse (if it is invertible). For a given complex number $c$, we use $\angle c$, $\textrm{Re}\{c\}$, and $\textrm{Im}\{c\}$ to denote its phase, real part, and imaginary part, respectively. We use \textcolor{black}{$\mathbf{x}_m$} and \textcolor{black}{$\left[\mathbf{x}\right]_m$} to denote the \textcolor{black}{$m$-th} component of the vector $\mathbf{x},$ $\mathbf{x}_{n:m}$ (with $n<m$) to denote the vector formed by components of $\mathbf{x}$ from index $n$ to index $m$, $\textrm{Diag}(\mathbf{\mathbf{x}})$ to denote the diagonal matrix formed by all components of $\mathbf{x},$ and $\|\mathbf{x}\|$ to denote the Euclidean norm of the vector $\mathbf{x}$. We use $\mathbb{E}_{\mathbf{w}}\{\cdot\}$ to represent the expectation operation with respect to the random variable $\mathbf{w}$. Furthermore, the \textcolor{black}{notation} $\mathbf{1}$ ($\mathbf{1}_M$) and $\mathbf{I}$ ($\mathbf{I}_M$) represent the all-one vector and the identity matrix (of an appropriate size/size $M$), respectively. \textcolor{black}{And} the set of $M$-dimensional real (complex) vectors is denoted by $\mathbb{R}^{M}$ ($\mathbb{C}^{M}$) and the set of $M\times M$ Hermitian matrices is denoted by $\mathbb{H}^{M}.$ \textcolor{black}{Finally, we use $M$ and $K$ to denote the number of sensors and the number of measurements, respectively.} 

\section{The Asynchronous Multi-Sensor Registration Problem}

Consider a multi-sensor system consisting of $M>1$ sensors located distributively on a 2-dimensional plane with known positions.
Suppose that there is a reference target (e.g., a civilian airplane) \textcolor{black}{moving in the surveillance region,}
and sensors measure the relative range and azimuth between the target and sensors themselves in an asynchronous work mode.
For ease of {notation} and presentation, measurements from different sensors are mapped onto a common time axis at the fusion center, indexed by $k$.
Furthermore, we assume that, at time instance $k,$ only one sensor observes the target and the corresponding sensor is denoted as $s_k\in\{1,2,\ldots,M\}.$ See Fig.$\ $\ref{fig:AcyTime} for an illustration of the asynchronous work mode with $M=2$ sensors.
\begin{figure}[h]
	\centering
	\includegraphics [width=0.6 \linewidth]{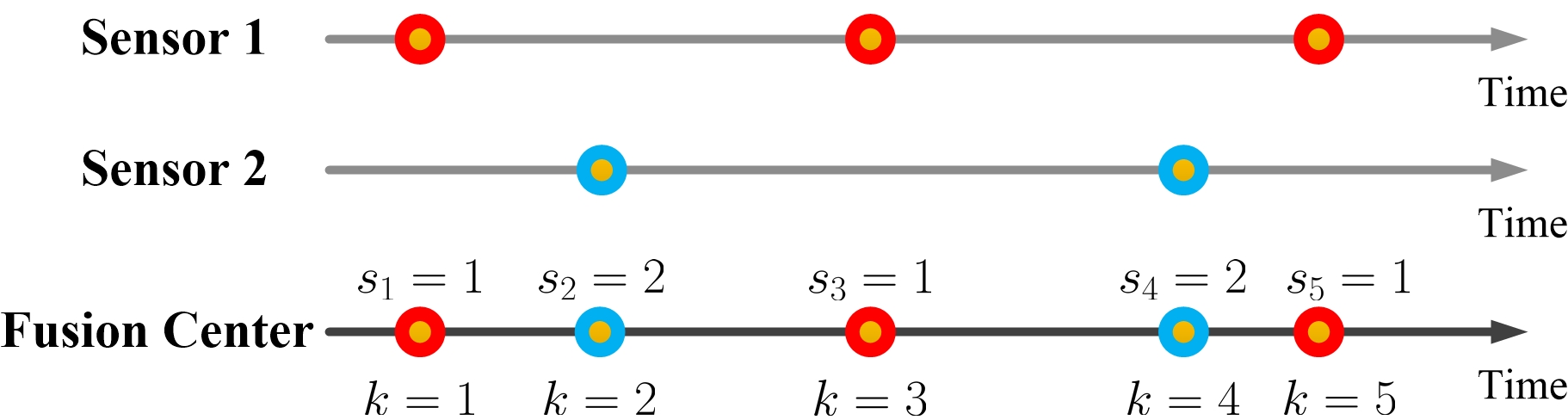}
	\caption{ \small An illustration of the asynchronous work mode with 2 sensors.}
	\label{fig:AcyTime}
\end{figure}

Let $\bm{\xi}_k=[x_k,y_k]^T$ denote the target position in the common coordinate system at time instance $k.$ Then, the measured range $\rho_k$ and azimuth $\phi_k$ at sensor $s_k$ are
\begin{equation}\label{eq:MeasureModel}
\mathbf{z}_k\triangleq
[\rho_k,
\phi_k
]^T=
{h}^{-1}(\bm{\xi}_k-\mathbf{p}_{s_k})-\bm{\bar{\theta}}_{s_k}+\mathbf{w}_k.
\end{equation}
In the above, $h^{-1}(\cdot)$ is the inverse function of the 2-dimensional spherical-to-Cartesian transformation function $h(\cdot)$;
$\mathbf{p}_m=[p_m^{x},p_m^{y}]^T$ is the position of sensor $m$; $\bm{\bar{\theta}}_{m}=[\Delta\bar{\rho}_{m},\Delta\bar{\phi}_{m}]^T$ is the true range and azimuth biases of sensor $m$; $\mathbf{w}_k=[w^{\rho}_{k}, w^{\phi}_{k}]^T$ is an uncorrelated Gaussian random noise vector with zero mean \cite{li2001survey}, i.e.,
$\mathbf{w}_k\sim\mathcal{N}\left(\mathbf{0},\textrm{diag}([\sigma^{2}_{\rho},\sigma^{2}_{\phi}]) \right)$ where $\sigma_{\rho}$ and $\sigma_{\phi}$ are the standard deviation of the noise in range and azimuth measure, respectively.
The biased measurements for a target with {a} constant velocity $\mathbf{\bar{v}}=[\bar{v}_x,\bar{v}_y]^T\neq\mathbf{0}$ are illustrated in Fig. \ref{fig:AcySC}.
\begin{figure}[h]
	\centering
	\includegraphics [width=0.6 \linewidth]{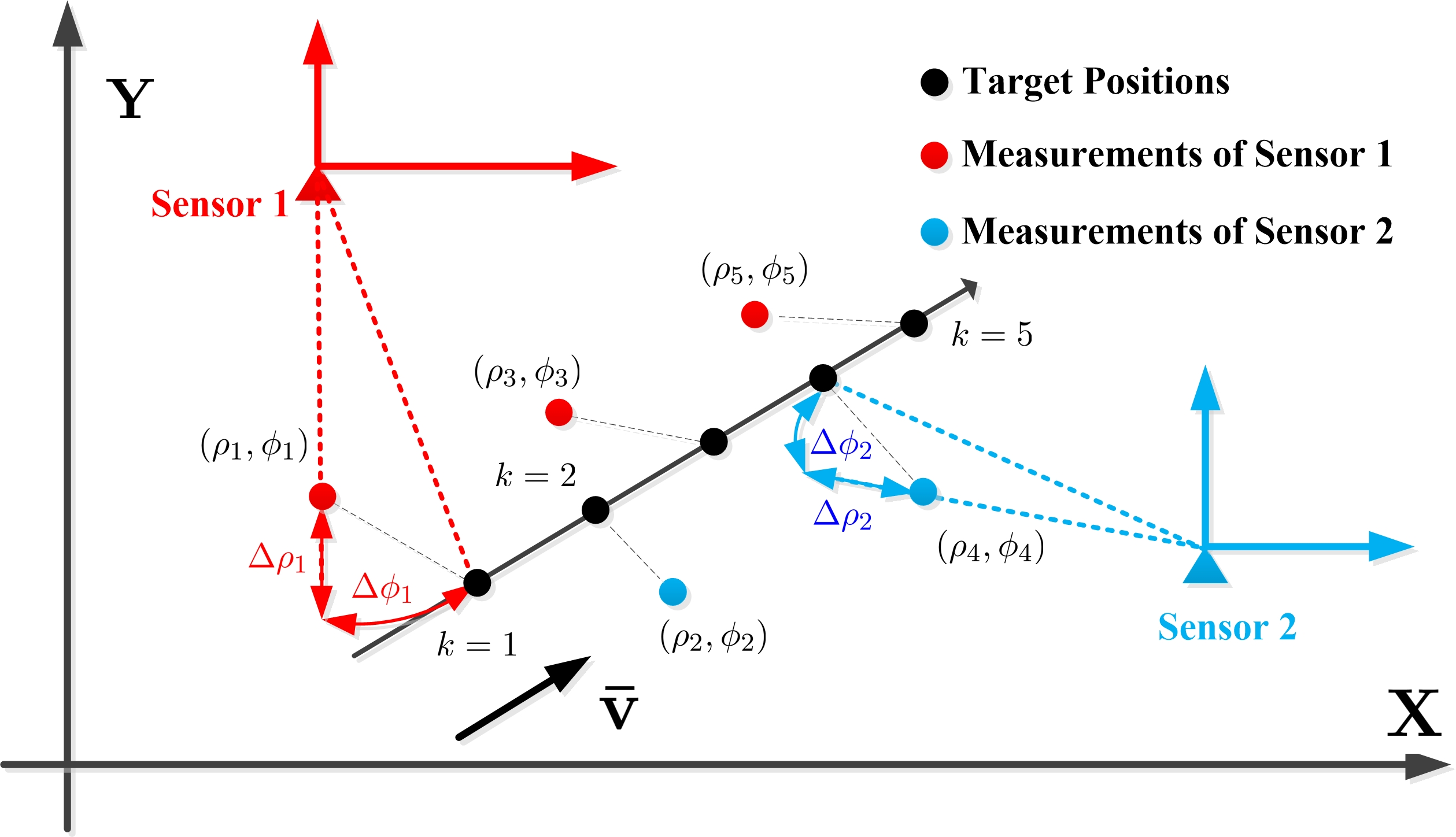}
	\caption{ \small An illustration of biased measurements of 2 sensors for {a} target with a constant velocity $\mathbf{\bar{v}}$.}
	\label{fig:AcySC}
\end{figure}

The asynchronous multi-sensor registration problem considered in this paper is to estimate sensor biases $\{\bm{\bar{\theta}}_m\}_{m=1}^M$ from the noisy measurements $\{\mathbf{z}_k\}_{k=1}^K$, where $K$ is the total number of measurements.
\textcolor{black}{Notice that there are $2K+2M$ unknown parameters $\{\bm{\xi}_k\}_{k=1}^K$ and $\{\bm{\bar{\theta}}_{m}\}_{m=1}^M$ but in total only $2K$ measurements $\{\mathbf{z}_k\}_{k=1}^K$ and thus the estimation problem is underdetermined.
	Assuming that the target moves with a \textcolor{black}{nearly constant velocity} and by exploiting this side information, we can estimate $\{\bm{\bar{\theta}}_m\}_{m=1}^M$ from measurements $\{\mathbf{z}_k\}_{k=1}^K$.}


\section{A Nonlinear Least Squares Formulation}
{\color{black}{In this section, we first introduce a 2-dimensional target motion model and then propose a nonlinear {LS} formulation for the sensor registration problem.
		Specifically, we assume that the target motion model follows a {nearly-constant-velocity model} \cite{li2003survey}, which is widely used for modeling the target motion with an approximated constant velocity:
		\begin{subequations}\label{eq:motion}
			\begin{align}
			\bm{\xi}_{k+1}&=\bm{\xi}_{k}+T_k\bm{\dot{\xi}}_{k}+\mathbf{n}_k,\label{eq:pos}\\
			\bm{\dot{\xi}}_{k+1}&=\bm{\dot{\xi}}_{k} + \mathbf{\dot{n}}_k,\label{eq:vel}
			\end{align}
		\end{subequations}
		where $T_k$ is the time difference between time instances $k$ and $k+1$, $\bm{\xi}_{k}\in\mathbb{R}^2$ and $\bm{\dot{\xi}}_{k}\in\mathbb{R}^2$ are the target position and velocity at time instance $k$ respectively, $\mathbf{n}_k\in\mathbb{R}^2$ and $\mathbf{\dot{n}}_k\in\mathbb{R}^2$ are the zero-mean motion process noise for the target position and velocity at time instance $k$. The covariance of the motion process noise is defined as
		$$
		\mathbf{Q}_k\triangleq\mathbb{E}\{ [\mathbf{n}_k^T,\mathbf{\dot{n}}_k^T]^T[\mathbf{n}_k^T,\mathbf{\dot{n}}_k^T]\}=q\begin{bmatrix}
		T_k^3/3&T_k^2/2\\
		T_k^2/2&T_k
		\end{bmatrix}\otimes\mathbf{I}_2 ,
		$$
		where $\otimes$ denotes the Kronecker operation and $q>0$ is the noise density which captures the magnitude of the uncertainty level in the target motion. Furthermore, assume the initial target states obey the Gaussian distribution, modeled as
\begin{equation}\label{barv}\bm{\xi}_1=\bm{\bar{\xi}}+\mathbf{n}_0,\ \bm{\dot{\xi}}_1=\mathbf{\bar{v}}+\mathbf{\dot{n}}_0,\end{equation}
		where $\bm{\bar{\xi}}$ and $\mathbf{\bar{v}}$ are the expected initial position and velocity, and $\mathbf{n}_0$ and $\mathbf{\dot{n}}_0$ are the zero-mean Gaussian noise for the initial position and velocity respectively. The covariance matrices of $\mathbf{n}_0$ and $\mathbf{\dot{n}}_0$ are assumed to be $c\mathbf{I}_2$ and $\dot{c}\mathbf{I}_2$, where $c\geq0$ and $\dot{c}\geq0$ are some constants that capture the magnitude of the uncertainty of the target's initial position and velocity.

Based on the measurement model \eqref{eq:MeasureModel} and the target motion model \eqref{eq:motion}, we have the following proposition which relates the asynchronous measurements $\mathbf{z}_{k+1}$ and $\mathbf{z}_k$.
		\begin{prop}\label{prop:ls}
			Suppose that the measurement model and the target motion model are given in \eqref{eq:MeasureModel} and \eqref{eq:motion}, respectively.
		Then
		\begin{equation}\label{eq:exp_model2}
		\begin{aligned}
		&\mathbb{E}_{
		\mathbf{n}_0,\mathbf{\dot{n}}_0,\mathbf{n}_{1},\mathbf{\dot{n}}_1,\ldots,\mathbf{n}_{k-1},\mathbf{\dot n}_{k-1},\mathbf{{n}}_{k}}\left\{ \mathbb{E}_{\mathbf{w}_{k+1}}\left\{ \bar{h}(\mathbf{z}_{k+1}+\bm{\bar{\theta}}_{s_{k+1}})\right\} \right\}+\mathbf{p}_{s_{k+1}}\\
		&=
		\mathbb{E}_{
		\mathbf{n}_0,\mathbf{\dot{n}}_0,\mathbf{n}_{1},\mathbf{\dot{n}}_1,\ldots,\mathbf{n}_{k-1}}\left\{ \mathbb{E}_{\mathbf{w}_k}\left\{ \bar{h}(\mathbf{z}_{k}+\bm{\bar{\theta}}_{s_{k}})\right\} \right\}+\mathbf{p}_{s_{k}}+T_k\mathbf{\bar{v}},
		\end{aligned}
		\end{equation}
where $\bar{h}(\cdot)$ is defined as
			\begin{equation}\label{EQ_h}
			\bar{{h}}(\rho,\phi)\triangleq\begin{bmatrix}\lambda^{-1}\rho\cos\phi\\ \lambda^{-1}\rho\sin\phi\end{bmatrix}~\text{and}~\lambda\triangleq e^{-\sigma_{\phi}^2/2}.
			\end{equation}	
\end{prop}

\begin{proof}
		First of all, it follows from \eqref{eq:motion} that, for each $k\geq 1,$ $\bm{\xi}_{k+1}$ depends on $\mathbf{n}_0,\mathbf{\dot{n}}_0,\mathbf{n}_1,\mathbf{\dot{n}}_1,\ldots,\mathbf{n}_{k-1},\mathbf{\dot{n}}_{k-1},\mathbf{n}_k$ and $\bm{\dot \xi}_{k+1}$ depends on $\mathbf{\dot{n}}_0,\mathbf{\dot{n}}_1,\ldots,\mathbf{\dot{n}}_k.$  Then, for each $k\geq 1,$ taking expectation of $\bm{\xi}_{k+1}$ and $\bm{\dot{\xi}}_{k+1}$ with respect to $\mathbf{n}_i$ and $\mathbf{\dot{n}}_i$ for all $i=1,2,\ldots,k$, we get
		\begin{subequations}\label{eq:exp_motion}
			\begin{align}
			\mathbb{E}_{
			\mathbf{n}_0,\mathbf{\dot{n}}_0,\mathbf{n}_{1},\mathbf{\dot{n}}_1,\ldots,\mathbf{n}_{k-1},\mathbf{\dot{n}}_{k-1},\mathbf{n}_k}\left\{\bm{\xi}_{k+1}\right\}
		&=\mathbb{E}_{
			\mathbf{n}_0,\mathbf{\dot{n}}_0,\mathbf{n}_{1},\mathbf{\dot{n}}_1,\ldots,\mathbf{n}_{k-1}}\left\{\bm{\xi}_{k}\right\}\nonumber\\
             &~~+T_k\mathbb{E}_{\mathbf{\dot{n}}_0,\mathbf{\dot{n}}_{1},\ldots,\mathbf{\dot{n}}_{k-1}}\left\{\bm{\dot{\xi}}_{k}\right\}+\mathbb{E}_{\mathbf{{n}}_k}\left\{\bf{{n}}_{k}\right\},\label{eq:exp_pos}\\
			\mathbb{E}_{
				\mathbf{\dot{n}}_0,\mathbf{\dot{n}}_{1},\ldots,\mathbf{\dot{n}}_k}\{\bm{\dot{\xi}}_{k+1}\}
			&=\mathbb{E}_{
				\mathbf{\dot{n}}_0,\mathbf{\dot{n}}_{1},\ldots,\mathbf{\dot{n}}_{k-1}}\{\bm{\dot{\xi}}_{k}\}+\mathbb{E}_{\mathbf{\dot{n}}_k}\left\{\bm{\dot{n}}_{k}\right\}. \label{eq:exp_vel} 
			\end{align}
		\end{subequations}
		Combining \eqref{barv}, \eqref{eq:exp_pos}, and \eqref{eq:exp_vel} yields
		\begin{equation}
		\begin{aligned}
		&\mathbb{E}_{
			\mathbf{n}_0,\mathbf{\dot{n}}_0,\mathbf{n}_{1},\mathbf{\dot{n}}_1,\ldots,\mathbf{n}_{k-1},\mathbf{\dot n}_{k-1},\mathbf{n}_k}\{\bm{\xi}_{k+1}\}=\mathbb{E}_{
			\mathbf{n}_0,\mathbf{\dot{n}}_0,\mathbf{n}_{1},\mathbf{\dot{n}}_1,\ldots,\mathbf{n}_{k-1}}\{\bm{\xi}_{k}\}
+T_k\mathbf{\bar{v}}.\label{eq:exp_model}
		\end{aligned}
		\end{equation}
Moreover, the measurement model \eqref{eq:MeasureModel} can be rewritten as $${h}(\mathbf{z}_{k}-\mathbf{w}_{k}+\bm{\bar{\theta}}_{s_{k}})=\bm{\xi}_k-\mathbf{p}_{s_{k}},$$
		which, together with the definition of $\bar{h}(\cdot)$ in \eqref{EQ_h} \cite{xiaoquan1998unbiased}, further implies
		\begin{equation}\label{eq:exp_z}
		 \mathbb{E}_{\mathbf{w}_k}\{\bar{{h}}(\mathbf{z}_{k}+\bm{\bar{\theta}}_{s_{k}})\}={h}\left(\mathbb{E}_{\mathbf{w}_k}\{\mathbf{z}_{k}+\bm{\bar{\theta}}_{s_{k}}\}\right)=\bm{\xi}_k-\mathbf{p}_{s_{k}}.
		\end{equation}
		Now, we can see from \eqref{eq:exp_model} and \eqref{eq:exp_z} that the measurements $\mathbf{z}_{k+1}$ and $\mathbf{z}_k$ satisfies the desired equation \eqref{eq:exp_model2}.  \qed
		\end{proof}

}}

\textcolor{black}{From Proposition \ref{prop:ls}}, we can formulate the problem of estimating sensor range biases $\bm{\Delta\rho}=[\Delta\rho_1,\Delta\rho_2,\ldots,\Delta\rho_M]^T$, azimuth biases $\bm{\Delta\phi}=[\Delta\phi_1,\Delta\phi_2,\ldots,\Delta\phi_M]^T$, and the expected target velocity $\mathbf{v}$ as the following nonlinear LS problem
\begin{equation}\label{eq:multi_optori}
\min_{ \bm{\Delta\rho}, \bm{\Delta\phi}, \mathbf{v} } \sum_{k=1}^{K-1} \| {g}_{k+1}\left( \bm{\theta}_{s_{k+1}}\right)-{g}_k\left(\bm{\theta}_{s_{k}}\right)-T_k\mathbf{v}\|^2,
\end{equation}
where
\begin{align*}
{g}_k(\bm{\theta}_{s_k})&=\bar{{h}}\left(\mathbf{z}_{k}+\bm{\theta}_{s_k}\right)+\mathbf{p}_{s_k},\  k=1,2,\ldots,K,\\ \bm{\theta}_{s_{k}}&=[\Delta\rho_{s_k},\Delta\phi_{s_k}]^T,\  k=1,2,\ldots,K.
\end{align*}

Problem \eqref{eq:multi_optori} is a non-convex problem due to the nonlinearity of $\bar{h}(\cdot)$. \textcolor{black}{In the sequel, we shall develop efficient algorithms to solve problem \eqref{eq:multi_optori} by exploiting its special structure.} In particular, we first study the single-sensor case of problem \eqref{eq:multi_optori} with $M=1$ in Section \ref{sec:single} and then study the multi-sensor case of problem \eqref{eq:multi_optori} with $M\geq 2$ in Section \ref{sec:multi}.

\section{\textcolor{black}{Single-Sensor Case: Separation Property}}
\label{sec:single}
Consider problem \eqref{eq:multi_optori} with $M=1$. With a slight abuse of \textcolor{black}{notation}, we still use $\{\mathbf{z}_k\}_{k=1}^K$ to denote the measurements from the sensor and use $\bm{\theta}=[\Delta \rho,\Delta \phi]^T$ to denote the range and azimuth biases to be estimated.

Problem \eqref{eq:multi_optori} with only one sensor can be simplified to:
\begin{equation}\label{EQ_OptSingle}
\min_{\bm{\theta},\mathbf{v}}f(\bm{\theta},\mathbf{v})\triangleq\sum_{k=1}^{K-1} \| \bar{{h}}(\mathbf{z}_{k+1}+\bm{\theta})-\bar{{h}}(\mathbf{z}_k+\bm{\theta})-T_k\mathbf{v}\|^2 .
\end{equation}
In light of \eqref{EQ_h}, $f(\bm{\theta},\mathbf{v})$ is a convex quadratic function with respect to $\Delta \rho$ and $\mathbf{v}$, \textcolor{black}{but is nonlinear and non-convex in terms of $\Delta\phi$.}
For any fixed $\Delta\phi$, problem \eqref{EQ_OptSingle} can be equivalently rewritten as the following convex quadratic program with respect to $\Delta\rho$ and $\mathbf{v}$:

\begin{equation}\label{eq:single_fixphi}
\min_{\Delta\rho,\mathbf{v}}\quad  \|\mathbf{H}_{0}\Delta\rho+\mathbf{H}_{1}\mathbf{v}-\mathbf{y}\|^2,
\end{equation}
where
\begin{equation}\label{eq:single_Handy}
\begin{aligned}
\mathbf{H}_{0}&=
[
c_1,
s_1,
c_2,
s_2,
\cdots,
c_{K-1},
s_{K-1}
]^T\in\mathbb{R}^{2(K-1)},
\\
\mathbf{H}_{1}&=
[-T_1,-T_2,\cdots,-T_{K-1}]^T\otimes \mathbf{I}_2\in\mathbb{R}^{2(K-1)\times 2},\\
\mathbf{y}&=
\left[
y_1^c,y_1^s,y_2^c,y_2^s, \ldots,y_{K-1}^c,y_{K-1}^s
\right]^T\in\mathbb{R}^{2(K-1)}.
\end{aligned}
\end{equation}
In the above,
\begin{align}
c_k=\lambda^{-1}\left[\cos(\phi_{k+1}+\Delta\phi)-\cos(\phi_{k}+\Delta\phi) \right],\label{eq:c_k}
\\
s_k=\lambda^{-1}\left[\sin(\phi_{k+1}+\Delta\phi)-\sin(\phi_{k}+\Delta\phi) \right],\label{eq:s_k}
\end{align}
and
{
\begin{align}
y_k^c=-\lambda^{-1}\left[\rho_{k+1}\cos(\phi_{k+1}+\Delta\phi)-\rho_{k}\cos(\phi_{k}+\Delta\phi) \right],\label{eq:yc_k}
\\
y_k^s=-\lambda^{-1}\left[\rho_{k+1}\sin(\phi_{k+1}+\Delta\phi)-\rho_{k}\sin(\phi_{k}+\Delta\phi) \right].\label{eq:ys_k}
\end{align}
}
Suppose $K\geq3$ (such that $\mathbf{H}\triangleq[\mathbf{H}_1,\mathbf{H}_2]$ is generically full column rank\footnote{\textcolor{black}{There exist degenerate scenarios where $\textrm{rank}(\mathbf{H})< 3$ even when $K\geq3$, e.g., the target trajectory and the sensor are on the same straight line. In this paper, we will focus on generic scenarios where $\text{rank}(\mathbf{H})=3$ when $K\geq3$.}}), then the unique optimal solution of problem (\ref{eq:single_fixphi}) is \cite{boyd2004convex}
\begin{equation}\label{EQ_SoluRho}
[
\Delta\rho^*,
\mathbf{v}^*
]^T=\left(\mathbf{H}^T\mathbf{H}\right)^{-1}\mathbf{H}^T\mathbf{y}.
\end{equation}
The following Theorem \ref{thm:single} shows, somewhat surprisingly, that $\Delta \rho^*$ in \eqref{EQ_SoluRho} does not depend on $\Delta \phi$ and hence is optimal to problem \eqref{EQ_OptSingle}.
\begin{thm}\label{thm:single}
	For any given $\Delta \phi$, problems \eqref{EQ_OptSingle} and \eqref{eq:single_fixphi} have the same optimal $\Delta \rho^*$ given by \eqref{EQ_SoluRho} and the same optimal objective value. However, the optimal $\mathbf{v}^*$ in \eqref{EQ_SoluRho} depends on $\Delta \phi.$
\end{thm}

The proof of Theorem \ref{thm:single} is given in Appendix A. Here we give an intuitive explanation of Theorem \ref{thm:single} \textcolor{black}{using} Fig.$\ $\ref{fig:geo_single} below. Suppose that there is no motion process noise or measurement noise.
	Given the original measurements $\mathbf{z}_1,\mathbf{z}_2,\mathbf{z}_3$ (green points), problem \eqref{EQ_OptSingle} aims at finding an azimuth bias $\Delta \phi,$ a range bias $\Delta \rho,$ and a velocity vector $\mathbf{v}$ to minimize the matching errors (corresponding to the square sum of the length of those black segments in Fig.$\ $\ref{fig:geo_single}).
	As shown in Fig.$\ $\ref{fig:geo_single}, when we rotate green points to blue points by $\Delta \phi$ or to blue circles by $\Delta \phi',$ the relative positions between the red and orange points (circles) do not change. In other words, \textcolor{black}{the optimal} $\Delta \rho$ and the optimal value of problem \eqref{EQ_OptSingle} are the same for both $\Delta \phi$ and $\Delta \phi'$.
	However, the optimal velocity of problem \eqref{EQ_OptSingle} indeed changes, i.e., it changes from $\mathbf{v}$ to $\mathbf{v}^{\prime}$ when the rotation angle changes from $\Delta \phi$ to $\Delta \phi^{\prime}.$

\begin{figure}[H]
	\centering
	\includegraphics [width=3.5in]{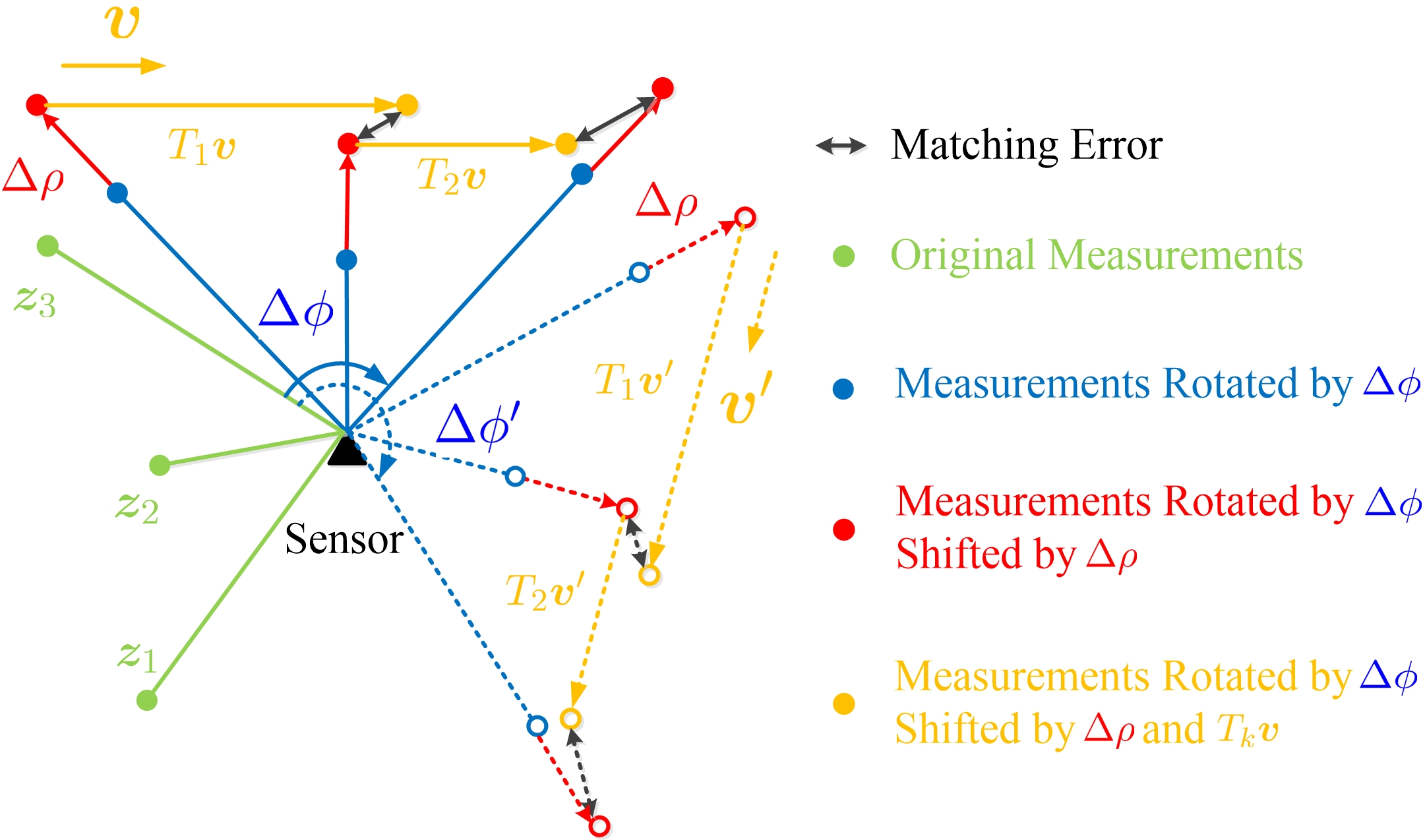}
	\caption{ \small A geometrical explanation of Theorem \ref{thm:single}. }
	\label{fig:geo_single}
\end{figure}

Based on Theorem \ref{thm:single}, we have the following corollary.

{\color{black}{
		\begin{coro}\label{coro:single}
			If there is no noise, i.e., $\mathbf{n}_k=\mathbf{\dot{n}}_k=\mathbf{0}$ for all $k$ and $\sigma_\rho^2=\sigma_\phi^2=0$, then solving problem \eqref{eq:single_fixphi} can recover the true range bias, $\Delta\rho^*=\Delta\bar{\rho}$, where $\Delta\bar{\rho}$ is the true range bias.
		\end{coro}
}}
\begin{proof}
	\textcolor{black}{In the absence of the motion process noise and measurement noise}, we have, from measurement model \eqref{eq:MeasureModel}, target motion model \eqref{eq:motion}, and definition of $\bar{h}(\cdot)$, that
	$$\bar{{h}}(\mathbf{z}_{k+1}+\bm{\bar{\theta}})-\bar{{h}}(\mathbf{z}_k+\bm{\bar{\theta}})-T_k\mathbf{\bar{v}}=\mathbf{0}.$$
	In other words, the optimal value of problem \eqref{EQ_OptSingle} is zero and thus {{$\bm{\bar{\theta}}$}} is its global minimizer. Suppose that $\Delta\phi$ in problem \eqref{eq:single_fixphi} is the true azimuth bias $\Delta\bar{\phi}$. Then, $\Delta\rho^*$ in \eqref{EQ_SoluRho} should be equal to $\Delta\bar{\rho}$. By Theorem \ref{thm:single}, $\Delta\rho^*$ is independent of $\Delta\phi$. Therefore, for any fixed $\Delta\phi$,  $\Delta\rho^*=\Delta\bar{\rho}$ always holds. \qed
\end{proof}

Theorem \ref{thm:single} tells us that each sensor can estimate its range bias $\Delta \rho$ independently by solving problem \eqref{EQ_OptSingle} (or problem \eqref{eq:single_fixphi}). Moreover, \textcolor{black}{in the absence of motion process noise and measurement noise,} the range bias can be exactly recovered by \eqref{EQ_SoluRho}, as shown in Corollary \ref{coro:single}. However, each sensor cannot estimate its azimuth bias $\Delta \phi$ and target velocity $\mathbf{v}$ independently by solving problem \eqref{EQ_OptSingle} due to the ambiguity of $\Delta \phi$ and $\mathbf{v}$ in problem \eqref{EQ_OptSingle}, i.e., problem \eqref{EQ_OptSingle} has multiple optimal pairs $(\Delta \phi, \mathbf{v})$. The ambiguity of $\Delta \phi$ and $\mathbf{v}$ arising in the single-sensor case can be solved by combining measurements from different sensors. In particular, both range and azimuth biases as well as target velocity can be estimated by solving our proposed nonlinear LS formulation \eqref{eq:multi_optori} as shown in Section \ref{sec:multi}.

We conclude this section by using Fig. \ref{fig:geo_amg}, which illustrates how the ambiguity of {{$\Delta \phi$}} and $\mathbf{v}$ in the single-sensor case can be resolved by combining measurements from two different sensors. \textcolor{black}{Assume that all kinds of noises are absent}, by solving problem \eqref{EQ_OptSingle} (or problem \eqref{eq:single_fixphi}), each sensor can recover its true range bias. \textcolor{black}{In other words,} their local measurements can be aligned onto a straight line (corresponding to the black dash lines in Fig. \ref{fig:geo_amg}) by compensating range biases. Since each sensor's measurements are from one target with a constant velocity, there is only one possibility to rotate those dash lines onto a \textcolor{black}{common} line (corresponding to the black solid line in Fig. \ref{fig:geo_amg}). Therefore, there is no ambiguity of azimuth biases and target velocity in the two-sensor case.
\begin{figure}[h]
	\centering
	\includegraphics [width=0.6\linewidth]{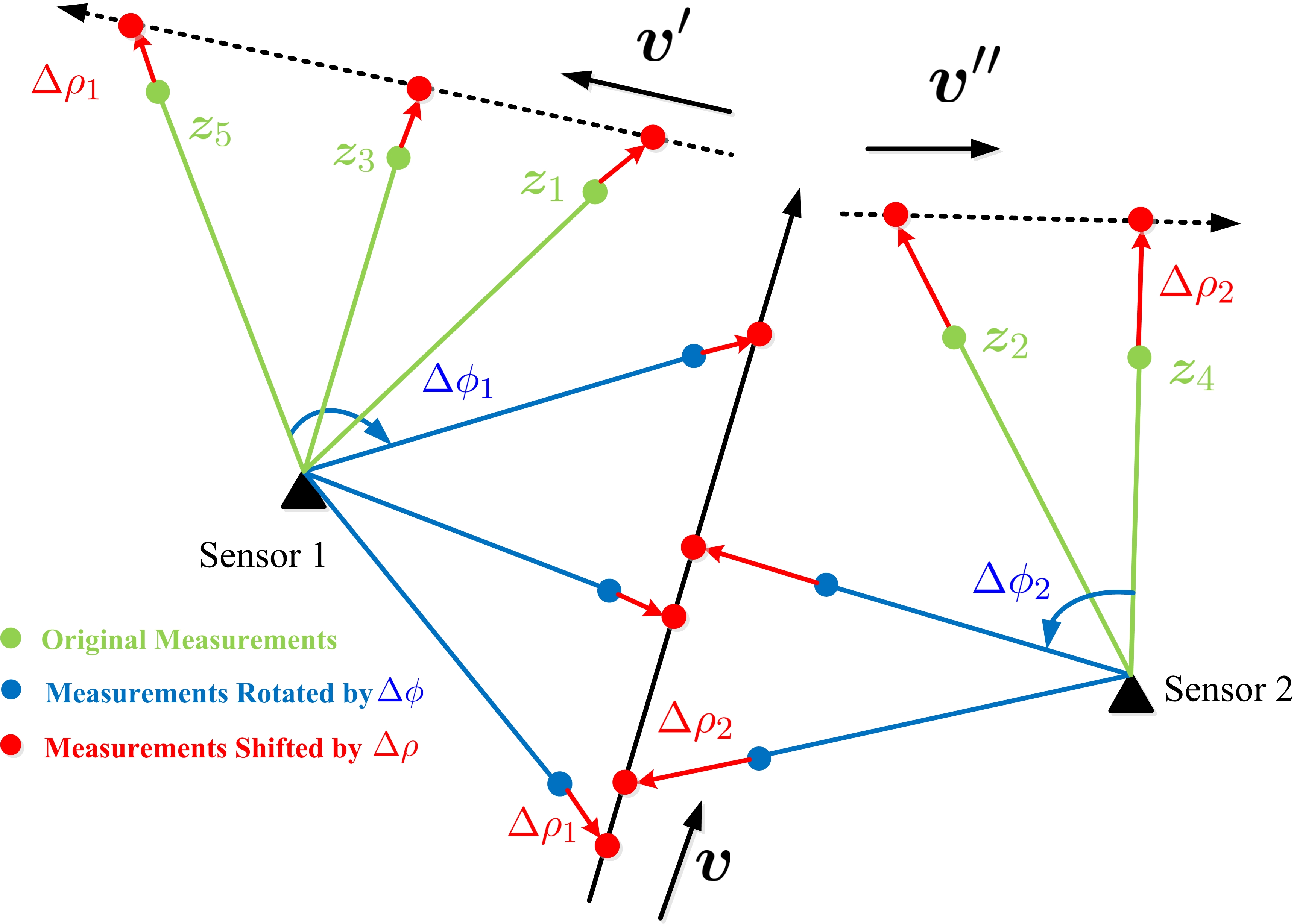}
	\caption{ \small A geometrical explanation of resolving ambiguity of azimuth biases and target velocity by combining measurements from 2 different sensors.}
	\label{fig:geo_amg}
\end{figure}

\section{\textcolor{black}{Multi-Sensor Case: A Block Coordinate Descent Algorithm}}
\label{sec:multi}
In this section, we study problem \eqref{eq:multi_optori} in the multi-sensor case, i.e., $M\geq 2$. In this case, the optimal $\bm{\Delta\rho}$ of problem \eqref{eq:multi_optori} depends on $\bm{\Delta\phi}$; see Eq. \eqref{eq:multi_solu_rho} further ahead. This is different from the case where the optimal $\Delta\rho$ in \eqref{EQ_SoluRho} does not depend on $\Delta\phi$ when there is only one sensor (see Theorem \ref{thm:single}).

The non-convexity of problem \eqref{eq:multi_optori} comes from the nonlinear terms $\Delta\rho_m\sin{\Delta\phi_m}$ and $\Delta\rho_m\cos{\Delta\phi_m}$. To handle such difficulty, we propose to use the {BCD algorithm} to alternately minimize $f(\bm{\Delta\rho}, \bm{\Delta\phi}, \mathbf{v})$ with respect to two blocks $\Delta\bm{\rho}$ and $(\Delta\bm{\phi},\mathbf{v})$ as follows:
\begin{subequations}
	\begin{equation}\label{eq:multi_bcd_pro1}
	\bm{\Delta\rho}^{t+1}=\arg\min_{\bm{\Delta\rho}}f(\bm{\Delta\rho}, \bm{\Delta\phi}^t, \mathbf{v}^t),
	\end{equation}
	\begin{equation}\label{eq:multi_bcd_pro2}
	(\bm{\Delta\phi}^{t+1},\mathbf{v}^{t+1})=\arg\min_{\bm{\Delta\phi},\mathbf{v}}f(\bm{\Delta\rho}^{t+1}, \bm{\Delta\phi}, \mathbf{v}),
	\end{equation}
\end{subequations}
	where $t\geq 1$ is the iteration index. \textcolor{black}{The (intuitive) reason for an alternate optimization of $\bm{\Delta\rho}$ and $(\bm{\Delta\phi},\mathbf{v})$ is that these two types of variables have different influences on the measurements as illustrated in Fig.  \ref{fig:geo_amg}, i.e., $\bm{\Delta\rho}$ `shifts' and $\bm{\Delta\phi}$ `rotates' the measurements. Another reason, from the mathematical point of view, is that both subproblems \eqref{eq:multi_bcd_pro1} and \eqref{eq:multi_bcd_pro2} can be solved globally and efficiently (under mild conditions), which will be shown in Sections \ref{subsec:multi_range} and \ref{subsec:multi_azi}, respectively. In Section \ref{subsec:multi_alg}, we give the BCD algorithm to solve problem \eqref{eq:multi_optori}.}

\subsection{Solution for Subproblem \eqref{eq:multi_bcd_pro1}}
\label{subsec:multi_range}
\textcolor{black}{For {notation}al simplicity, we omit the iteration index here. With fixed $(\bm{\Delta\phi},\mathbf{v})$, subproblem \eqref{eq:multi_bcd_pro1} can be reformulated as an unconstrained convex quadratic problem as follows:
	\begin{equation}\label{eq:multi_opt_rho}
	\min_{\bm{\Delta\rho}}\ \|\mathbf{G}\bm{\Delta\rho} -  \mathbf{y}\|^2.
	\end{equation}
	In the above, $\mathbf{y}= [\mathbf{y}_1^T,\mathbf{y}_2^T,\cdots,\mathbf{y}_{K-1}^T]^T\in\mathbb{R}^{2(K-1)}$, and
		\begin{equation}
		\begin{aligned}
		\mathbf{y}_k&=\textcolor{black}{-}[y_{k}^c+(p_{s_{k+1}}^x-p_{s_{k}}^x)-T_{k}v_x,
		y_{k}^s+(p_{s_{k+1}}^y-p_{s_{k}}^y)-T_{k}v_y]^T,\\
		y_k^c&=\lambda^{-1}\left[\rho_{k+1}\cos(\phi_{k+1}+\Delta\phi_{s_{k+1}})-\rho_k\cos(\phi_k+\Delta\phi_{s_k})\right],\\
		y_k^s&=\lambda^{-1}\left[\rho_{k+1}\sin(\phi_{k+1}+\Delta\phi_{s_{k+1}})-\rho_k\sin(\phi_k+\Delta\phi_{s_k})\right];
		\end{aligned}
		\end{equation}
	and the $(n,m)$-th entry of $\mathbf{G}\in\mathbb{R}^{2(K-1)\times M}$ is
	\begin{equation*}
	\left[\mathbf{G}\right]_{nm}=
	\left\lbrace
	\begin{aligned}
	&\lambda^{-1}\cos(\phi_{k+1}),&&\ \textrm{if $n$ is odd and $m=s_{k+1}$ },\\
	&-\lambda^{-1}\cos(\phi_{k}),&&\ \textrm{if $n$ is odd and $m=s_{k}$ },\\
	&\lambda^{-1}\sin(\phi_{k+1}),&&\ \textrm{if $n$ is even and $m=s_{k+1}$ },\\
	&-\lambda^{-1}\sin(\phi_{k}),&&\ \textrm{if $n$ is even and $m=s_{k}$ },\\
	&0, &&\  \textrm{otherwise},
	\end{aligned}
	\right.
	\end{equation*}
	where $k={\lfloor \frac{n} {2}\rfloor}$ and $\lfloor\cdot\rfloor$ denotes the floor operation. 
The closed-form solution of problem \eqref{eq:multi_opt_rho} is given by 
	\begin{equation}\label{eq:multi_solu_rho}
	\bm{\Delta\rho}^*=(\mathbf{G}^T\mathbf{G})^{-1}\mathbf{G}^T\mathbf{y}.
	\end{equation}} 
\subsection{Solution for Subproblem \eqref{eq:multi_bcd_pro2}}
\label{subsec:multi_azi}
To efficiently solve subproblem \eqref{eq:multi_bcd_pro2}, we first reformulate it as an equivalent complex quadratically constrained quadratic program (QCQP) as follows:
\begin{equation}
\begin{aligned}[c]\label{EQ_Reformulation_Complex}
\min_{\mathbf{x}\in\mathbb{C}^{M},v\in\mathbb{C}}\quad &\| \mathbf{H}\mathbf{x}-\mathbf{t}v + \mathbf{c}\|^2 \\
\textrm{s.t. } \quad\quad &|\mathbf{x}_{m}|^2 = 1,\quad \forall \  m,
\end{aligned}
\end{equation}
where $\mathbf{x}_{m}$ \textcolor{black}{is an equivalent representation of} the azimuth bias of sensor $m$ in the sense that $\Delta\phi_m=\angle \mathbf{x}_m$, and the complex scalar $v=v_x+jv_y$ represents the constant velocity.
In \eqref{EQ_Reformulation_Complex}, matrix $\mathbf{H}\in\mathbb{C}^{(K-1)\times M}$ are determined by sensor measurements $\{\mathbf{z}_k\}_{k=1}^K$ as follows:

\begin{equation}\label{eq:DefineH}
\left[\mathbf{H}\right]_{km}=
\left\lbrace
\begin{aligned}
&\lambda^{-1}(\rho_{k+1}+\Delta\rho_{s_{k+1}}) e^{j\phi_{k+1}},&&\ \textrm{if $m=s_{k+1}$},\\
&-\lambda^{-1}(\rho_{k}+\Delta\rho_{s_{k}}) e^{j\phi_{k}},&&\ \textrm{if $m=s_{k}$},\\
&0,&&\  \textrm{otherwise};
\end{aligned}
\right.
\end{equation}
vector $\mathbf{t}\in \mathbb{R}^{K-1}$ is related to time differences $\{T_k\}_{k=1}^{K-1}$ as follows:
\begin{equation*}
\mathbf{t}=[T_1,T_2,\ldots,T_{K-1} ]^T,
\end{equation*}
and $\mathbf{c}=[c_1,c_2,\ldots,c_{K-1}]^T\in\mathbb{C}^{K-1}$ is related to sensor positions $\{\mathbf{p}_m\}_{m=1}^M$ as follows:
\begin{equation*}
c_k=(p_{s_{k+1}}^x-p_{s_{k}}^x)+j(p_{s_{k+1}}^y-p_{s_{k}}^y),\  k=1,2,\ldots,K-1.
\end{equation*}

As an unconstrained quadratic program in $v$, problem \eqref{EQ_Reformulation_Complex} has a closed-form solution given by
\begin{equation}\label{EQ_Close_v}
v=(\mathbf{t}^\dagger\mathbf{t})^{-1}\mathbf{t}^\dagger(\mathbf{H}\mathbf{x}+\mathbf{c}).
\end{equation}
Plugging (\ref{EQ_Close_v}) into (\ref{EQ_Reformulation_Complex}), we get
\begin{equation}\label{EQ_OptDropv}
\begin{aligned}
\min_{\mathbf{x}\in\mathbb{C}^{M}}\quad&\|\mathbf{PH}\mathbf{x}+\mathbf{P}\mathbf{c}\|^2\\
\textrm{s.t. }\quad&|\mathbf{x}_m|^2=1,\  m=1,2,\ldots,M,
\end{aligned}
\end{equation}
where $\mathbf{P}=\mathbf{I}-\mathbf{t}\mathbf{t}^\dagger/\|\mathbf{t}\|^2$.

Problem (\ref{EQ_OptDropv}) is a non-convex QCQP, and \textcolor{black}{this} class of problems is known to be NP-hard in general \cite{luo20104}. One efficient convex relaxation technique to solve such class of problems, semidefinite relaxation (SDR), has shown its effectiveness in signal processing and communication communities \cite{luo2010semidefinite,lu2017tightness}.
We also apply the SDR technique to solve problem \eqref{EQ_OptDropv}. To do so, we further reformulate problem (\ref{EQ_OptDropv}) in a homogeneous form as follows:
\begin{equation}
\begin{aligned}\label{EQ_OptHomeQCQP}
\min_{\mathbf{x}\in\mathbb{C}^{M+1}}\quad&\mathbf{x}^\dagger\mathbf{C}\mathbf{x}\\
\textrm{s.t.}\quad\quad&|\mathbf{x}_m|^2=1,\  m=1,2,\ldots,M+1,
\end{aligned}
\end{equation}
where
\begin{equation*}
\mathbf{C}=\begin{bmatrix}\mathbf{H}^\dagger\mathbf{PH}&\mathbf{H}^\dagger\mathbf{P}\mathbf{c} \\ \mathbf{c}^\dagger\mathbf{PH}&0\end{bmatrix}.
\end{equation*}
It is simple to show that problems (\ref{EQ_OptDropv}) and (\ref{EQ_OptHomeQCQP}) are equivalent in the sense that $\mathbf{x}^*\in\mathbb{C}^{M+1}$ is the optimal solution for problem (\ref{EQ_OptHomeQCQP}) if and only if  ${\mathbf{x}^*_{1:M}}/{\mathbf{x}^*_{M+1}}\in\mathbb{C}^M$ is the optimal solution for problem (\ref{EQ_OptDropv}).

The SDP relaxation of problem \eqref{EQ_OptHomeQCQP} is
\begin{equation}
\begin{aligned}\label{eq:SDP}
\min_{\mathbf{X}\in\mathbb{H}^{M+1}}\quad&\textrm{Tr}(\mathbf{CX})\\
\textrm{s.t.}\quad\quad&\textrm{diag}(\mathbf{X})=\mathbf{1}, \\
&\mathbf{X}\succeq\mathbf{0}.
\end{aligned}
\end{equation}
Problem (\ref{eq:SDP}) can be efficiently solved by the interior-point algorithm \cite{helmberg1996interior}. If the optimal solution $\mathbf{X}^*$ for problem (\ref{eq:SDP}) is of rank one, i.e., $\mathbf{X}^*=\mathbf{x}^*(\mathbf{x}^*)^\dagger$, then the optimal solution for problem (\ref{eq:multi_bcd_pro2}) is obtained as follows:
\begin{equation}
\begin{aligned}\label{EQ_ExtracPhi}
\Delta\phi_m^*=\angle \frac{\mathbf{x}^*_m}{\mathbf{x}^*_{M+1}},\quad m=1,\ldots,M,
\end{aligned}
\end{equation}
and
\begin{equation}
\begin{aligned}\label{EQ_ExtracV}
\mathbf{v}^*=[\textrm{Re}\{ v^* \},\textrm{Im}\{ v^* \}]^T,\ v^* = \left(\mathbf{t}^\dagger\mathbf{t}\right)^{-1}\mathbf{t}^\dagger\left(\mathbf{H}\frac{\mathbf{x}^*_{1:M}}{\mathbf{x}^*_{M+1}}+\mathbf{c}\right).
\end{aligned}
\end{equation}
The dual problem of SDP \eqref{eq:SDP} is
\begin{equation*}
\begin{aligned}
\min_{\mathbf{y}\in\mathbb{R}^{M+1}}&\quad \mathbf{1}^T\mathbf{y}\\
\textrm{s.t.}\quad&\quad \mathbf{C}+\textrm{Diag}(\mathbf{y})\succeq \mathbf{0}.
\end{aligned}
\end{equation*}
The following lemma states sufficient conditions that SDP \eqref{eq:SDP} admits a unique rank-one solution.
\begin{lem}\label{lem:sufcond}
	SDP \eqref{eq:SDP} has a unique minimizer $\mathbf{X}$ of rank one if there exists $\mathbf{X}\in\mathbb{H}^{M+1}$ and $\mathbf{y}\in\mathbb{R}^{M+1}$ such that \\
	1. $\textrm{diag}(\mathbf{X})=\mathbf{1}$ and $\mathbf{X}\succeq\mathbf{0}$;\\
	2. $\mathbf{C}+\textrm{Diag}(\mathbf{y})\succeq\mathbf{0}$;\\
	3. $\left[ \mathbf{C}+\textrm{Diag}(\mathbf{y}) \right]\mathbf{X}=\mathbf{0}$; and\\
	4. $\mathbf{H}^\dagger\mathbf{P}\mathbf{H}+\textrm{Diag}(\mathbf{y}_{1:M})\succ\mathbf{0}$.
\end{lem}
\begin{proof}
	Notice that conditions 1, 2, and 3 are the Karush-Kuhn-Tucker (KKT) conditions of SDP \eqref{eq:SDP} \cite{boyd2004convex}. Since both primal and dual problems are strictly feasible, i.e., $\mathbf{I}_{M+1}$ is strictly feasible for the primal problem and $(|\lambda_{\textrm{min}}(\mathbf{C})|+\epsilon)\mathbf{1}$ with any $\epsilon>0$ is strictly feasible for the dual problem, it follows that Slater's condition holds true. Then, the KKT conditions 1, 2, and 3 are sufficient and necessary for optimality of primal and dual problems. Condition 4 implies $\mathbf{H}^\dagger\mathbf{P}\mathbf{H}+\textrm{Diag}(\mathbf{y}_{1:M})$ is nonsingular and thus $\textrm{rank}(\mathbf{C}+\textrm{Diag}(\mathbf{y}))\geq M$. This, together with  Sylvester’s rank inequality \cite{Horn1985matrix} and condition 3, immediately shows $\textrm{rank}(\mathbf{X})\leq 1$. Moreover, since $\mathbf{X}$ is non-zero (by condition 1), we have $\textrm{rank}(\mathbf{X})=1$ and $\textrm{rank}(\mathbf{C}+\textrm{Diag}(\mathbf{y}))=M$, which further implies $\mathbf{C}+\textrm{Diag}(\mathbf{y})$ is \textit{dual nondegenerate} \cite{Alizadeh1997Complementarity}. Therefore, it follows from Theorem 10 in \cite{Alizadeh1997Complementarity} that $\mathbf{X}$ is unique. \qed
\end{proof}

Lemma \ref{lem:sufcond} gives sufficient conditions on the existence and uniqueness of rank-one solution of SDP \eqref{eq:SDP}. However, these conditions are not always satisfied, because they indeed depend on the structure of $\mathbf{H}$, $\mathbf{P}$, $\mathbf{c}$, and the true azimuth biases {{$\bm{\Delta\bar{\phi}}=[\Delta\bar{\phi}_1,\Delta\bar{\phi}_2,\ldots,\Delta\bar{\phi}_M]^T$}} (as shown in Theorem \ref{thm:multi} later). In the following, we will present a mild condition such that SDP \eqref{eq:SDP} admits a unique minimizer of rank one. To begin with, we divide $\mathbf{H}$ in \eqref{EQ_Reformulation_Complex} into two parts as follows:
\begin{equation}\label{eq:divide_H}
\mathbf{H}=\mathbf{\tilde{H}}+\mathbf{\Delta H}.
\end{equation}
In \eqref{eq:divide_H}, $\mathbf{\tilde{H}}$ denotes the true part of $\mathbf{H}$ and is defined as
\begin{equation}\label{eq:tilde_H}
\left[\mathbf{\tilde{H}}\right]_{km}=
\left\lbrace
\begin{aligned}
&(\tilde{\rho}_{k+1}+\Delta\bar{\rho}_{s_{k+1}}) e^{j\tilde{\phi}_{k+1}},&&\ \textrm{if $m=s_{k+1}$},\\
&-(\tilde{\rho}_{k}+\Delta\bar{\rho}_{s_{k}}) e^{j\tilde{\phi}_{k}},&&\ \textrm{if $m=s_{k}$},\\
&0,&&\  \textrm{otherwise},
\end{aligned}
\right.
\end{equation}
where $\tilde{\rho}_{k}=\rho_{k}-w_{k}^{\rho}$ and $\tilde{\phi}_{k}=\phi_{k}-w_{k}^{\phi}$ ($w_{k}^{\rho}$ and $w_{k}^{\phi}$ are measurement noise defined in \eqref{eq:MeasureModel});
$\mathbf{\Delta H}=\mathbf{H}-\mathbf{\tilde{H}}$ represents \textcolor{black}{the error part of $\mathbf{H}$ caused by the motion process noise, the measurement noise, and possibly the inaccuracy of the fixed $\bm{\Delta\rho}$.
	Notice that if all kinds of noises are absent, and $\bm{\Delta\rho}$ in subproblem \eqref{eq:multi_bcd_pro2} (equivalent to \eqref{EQ_Reformulation_Complex}) is the true range biases {{$\bm{\Delta\bar{\rho}}=[\Delta\bar{\rho}_1,\Delta\bar{\rho}_2,\ldots,\Delta\bar{\rho}_M]^T$}}, then $\bm{\Delta}\mathbf{H}=\mathbf{0}$ and $\mathbf{H}=\mathbf{\tilde{H}}$.} In this case, SDP \eqref{eq:SDP} has the following exact recovery property.

\begin{coro}\label{coro:multi}
	Suppose $\bm{\Delta}\mathbf{H}=\mathbf{0}.$ Then, SDP \eqref{eq:SDP} always has a unique minimizer of rank one, i.e., $\mathbf{X}^*=\mathbf{x}^*(\mathbf{x}^*)^\dagger$. Furthermore, $\mathbf{x}^*$ exactly recovers the true azimuth biases
	{$$
		\Delta\bar{\phi}_m=\angle\frac{\mathbf{x}_m^*}{\mathbf{x}^*_{M+1}},\  m=1,2,\ldots,M.
		$$}
\end{coro}
\begin{proof}
	The proof consists of two parts. We first construct a pair of primal and dual solutions $\mathbf{X}^*$ and $\mathbf{y}^*$ and then prove their optimality and exact recovery property.
	
	Without loss of generality, we assume that each sensor has \textcolor{black}{sufficient} number of measurements such that $\mathbf{P}\mathbf{H}$ is of full column rank and $\mathbf{H}^\dagger\mathbf{P}^\dagger\mathbf{P}\mathbf{H}=\mathbf{H}^\dagger\mathbf{P}\mathbf{H}$ is invertible.
	Let
	\begin{equation}\label{eq:construct_x}
	\mathbf{X}^*=\mathbf{x}^*(\mathbf{x}^*)^\dagger,\ \textrm{where }\mathbf{x}^*=\begin{bmatrix}-(\mathbf{H}^\dagger\mathbf{P}\mathbf{H})^{-1}\mathbf{H}^\dagger\mathbf{P}\mathbf{c}\\ 1\end{bmatrix}\in\mathbb{C}^{M+1},
	\end{equation}
	and
	\begin{equation}\label{eq:construct_y}
	\mathbf{y}^*=\begin{bmatrix}\mathbf{0}\\ \mathbf{c}^\dagger\mathbf{P}^\dagger\mathbf{H}(\mathbf{H}^\dagger\mathbf{P}\mathbf{H})^{-1}\mathbf{H}^\dagger\mathbf{P}\mathbf{c}\end{bmatrix}\in\mathbb{R}^{M+1}.
	\end{equation}
	Notice that problem \eqref{EQ_OptDropv} is a reformulation of problem \eqref{eq:multi_optori} with fixed $\bm{\Delta\rho}$. In the absence of measurement noise, the optimal objective value of problem \eqref{eq:multi_optori} is zero and the optimal solutions are true sensor biases $\bm{\Delta\bar{\phi}}=[\Delta\bar{\phi}_1,\Delta\bar{\phi}_2,\ldots,\Delta\bar{\phi}_M]^T$ and true target velocity $\mathbf{\bar{v}}$. Therefore, if $\mathbf{H}=\mathbf{\tilde{H}}$, then the following equation holds true:
	\begin{equation}\label{eq:NoiseFree}
	\mathbf{P}\mathbf{\tilde{H}}[e^{j\bar{\phi}_1},e^{j\bar{\phi}_2}, \cdots, e^{j\bar{\phi}_M}]^T+\mathbf{P}\mathbf{c}=\mathbf{0}.
	\end{equation}
	Since $\mathbf{H}=\mathbf{\tilde{H}}$ and $\mathbf{PH}$ is of full column rank, it follows
	$$
	[e^{j\bar{\phi}_1},e^{j\bar{\phi}_2},\cdots , e^{j\bar{\phi}_M}]^T=-\left(\mathbf{H}^\dagger\mathbf{P}\mathbf{H}\right)^{-1}\mathbf{H}^\dagger\mathbf{P}\mathbf{c}.
	$$
	Consequently, $\mathbf{X}^*$ in \eqref{eq:construct_x} satisfies condition 1 in Lemma \ref{lem:sufcond}.
	
	Recall the definitions of $\mathbf{C}$ and $\mathbf{y}^*$. Then,
	$$
	\mathbf{C}+\textrm{Diag}(\mathbf{y}^*)=\begin{bmatrix}\mathbf{H}^\dagger\mathbf{PH}&\mathbf{H}^\dagger\mathbf{P}\mathbf{c} \\ \mathbf{c}^\dagger\mathbf{PH}&\mathbf{c}^\dagger\mathbf{P}^\dagger\mathbf{H}(\mathbf{H}^\dagger\mathbf{P}\mathbf{H})^{-1}\mathbf{H}^\dagger\mathbf{P}\mathbf{c}\end{bmatrix}.
	$$
	Since $\mathbf{H}^\dagger\mathbf{P}\mathbf{H}\succ \mathbf{0}$ and its Schur complement is zero, we know that $\mathbf{C}+\textrm{Diag}(\mathbf{y}^*)\succeq \mathbf{0}$, which shows that condition 2 in Lemma \ref{lem:sufcond} is true. Moreover, it is simple to check $\left[ \mathbf{C}+\textrm{Diag}(\mathbf{y}^*)\right]\mathbf{x}^*=\mathbf{0}$, which implies condition 3 in Lemma \ref{lem:sufcond}.  Since $\mathbf{H}^\dagger\mathbf{P}\mathbf{H}+\textrm{Diag}(\mathbf{y}^*_{1:M})=\mathbf{H}^\dagger\mathbf{P}\mathbf{H}\succ \mathbf{0}$, condition 4 in Lemma \ref{lem:sufcond} also holds. Therefore, the constructed solutions $\mathbf{X}^*$ in \eqref{eq:construct_x} and $\mathbf{y}^*$ in \eqref{eq:construct_y} satisfy all conditions in Lemma \ref{lem:sufcond}.  Hence $\mathbf{X}^*$ is the unique solution of SDP \eqref{eq:SDP} and  $\mathbf{x}^*$ exactly recovers the true azimuth biases. \qed
\end{proof}

Corollary \ref{coro:multi} shows that, if $\mathbf{\Delta H}=\mathbf{0}$ in \eqref{eq:divide_H}, the solution of SDP \eqref{eq:SDP} is rank one and exactly recovers the true azimuth biases. \textcolor{black}{The following Theorem \ref{thm:multi} generalizes this result to sufficient small $\mathbf{\Delta H} \neq \mathbf{0}$. In other words, if $\mathbf{\Delta H}$ is sufficiently small, SDP \eqref{eq:SDP} always admits a unique solution of rank one. The proof of Theorem \ref{thm:multi} is relegated to \textcolor{black}{Appendix B}. A sufficient condition on how $\mathbf{\Delta H}$ depends on the problem data and how small it should be to guarantee the unique rank-one solution of SDP \eqref{eq:SDP} is given in \textcolor{black}{\eqref{conditionpositive} \textcolor{black}{in Claim 2 in Appendix B}.}}

\begin{thm}\label{thm:multi}
	If $\mathbf{\Delta H}$ is sufficiently small, 
	then SDP \eqref{eq:SDP} admits a unique solution of rank one.
\end{thm}

\subsection{A BCD Algorithm}
\label{subsec:multi_alg}
\textcolor{black}{Now, we present a BCD algorithm  (Algorithm \ref{Alg} below) to solve the asynchronous multi-sensor registration problem \eqref{eq:multi_optori}. {The proposed BCD algorithm solves a quadratic program \eqref{eq:single_fixphi} (or \eqref{eq:multi_opt_rho}) and a SDP \eqref{eq:SDP} alternately. These steps represent the dominant per-iteration computational cost of the proposed BCD algorithm.} More specifically, the problem of estimating the range biases is equivalent to solving some linear equations, either $M$ of $3\times 3$ linear equations with positive definite coefficient matrices at iteration $t=0$ with a complexity of $\mathcal{O}(M)$ or a $M\times M$ linear equation with a positive definite coefficient matrix at $t\geq1$ with a complexity of $\mathcal{O}(M^3 )$; the worst-case complexity of solving a $M+1$ dimensional SDP in the form of \eqref{eq:SDP} is $\mathcal{O}(M^{4.5} )$ \cite{luo2010semidefinite}.}
\renewcommand{\algorithmicrequire}{\textbf{Input:}} 
\renewcommand{\algorithmicensure}{\textbf{Output:}} 
\begin{algorithm}[ht]
	\caption{A BCD Algorithm for Problem \eqref{eq:multi_optori}}
	\label{Alg}
	\begin{algorithmic}[1]
		\Require
		Measurements $\{\mathbf{z}_k\}_{k=1}^{K}$ collected by all sensors.
		\For{$t=0,1,2,\ldots,$}
		\If{$t=0$}
		\State Obtain $\Delta\rho_m^{t+1}$ by \eqref{EQ_SoluRho}, $m=1,2,\ldots,M$;
		\Else
		\State Obtain $\bm{\Delta\rho}^{t+1}$ by \eqref{eq:multi_solu_rho};
		\EndIf
		\State Solve SDP \eqref{eq:SDP} to obtain $\mathbf{X}^*$;
		\State \textcolor{black}{Extract $\{\Delta\phi_{m}^{*}\}_{m=1}^M$ and $\mathbf{v}^{*}$ by \eqref{EQ_ExtracPhi} and \eqref{EQ_ExtracV};}
		\EndFor
		\Ensure
		Estimated biases $\{\bm{\theta}_{m}^*\}_{m=1}^M$ and velocity $\mathbf{v}^*$.
	\end{algorithmic}
\end{algorithm}

The performance of the BCD algorithm generally depends on the choice of the initial point (especially when it is applied to solve the non-convex optimization problems). In our proposed BCD algorithm, we initialize it with \eqref{EQ_SoluRho} by using the separation property of the range bias estimation in the single-sensor case. Furthermore, global convergence of our proposed algorithm can be established by using similar arguments in \cite{grippo2000convergence}.
Based on Corollary \ref{coro:single} and Corollary \ref{coro:multi}, we further have the following result.
{{
		\begin{coro}\label{coro:exact}
			\textcolor{black}{If there is no noise, i.e., $\mathbf{n}_k=\mathbf{\dot{n}}_k=\mathbf{0}$ for all $k$ and $\sigma_\rho^2=\sigma_\phi^2=0$, Algorithm \ref{Alg} exactly recovers the true biases of all sensors and the true velocity of the target.}

		\end{coro}
}}
In the noiseless case, solving the sensor registration problem is equivalent to solving a set of nonlinear equations \eqref{eq:MeasureModel} and \eqref{eq:motion}. \textcolor{black}{Corollary \ref{coro:exact} shows that our proposed Algorithm \ref{Alg} is able to exactly solve these equations. This exact recovery property distinguishes our proposed algorithm from the existing approaches \cite{fischer1980registration,helmick1993removal,zhou1997an,lin2005multisensor,zhou2004asy,Ristic2003Sensor,Fortunati2011Least,Fortunati2013On,Dana1990Registration,leung1994least}, which use the first-order approximation to handle the nonlinearity in the sensor registration problem and cannot recover the sensor biases even if all kinds of noises are absent.}

{\color{black}{
		\section{Numerical Simulation }
		\label{sec:numexample}
		In this section, we evaluate the estimation performance of the proposed approach and compare it with five other approaches in the literature. The first approach is the recently proposed two-stage approach in \cite{pu2017two}, which is a special case of the proposed BCD approach with only one iteration. The second approach is a linearized LS approach for problem \eqref{eq:multi_optori} which can be regarded as a direct extension of approaches in \cite{Dana1990Registration,Fortunati2011Least,leung1994least} to solve our proposed nonlinear LS formulation \eqref{eq:multi_optori}. The third one is the approach proposed in \cite{lin2005multisensor,Lin2006Multisensor}, which smartly constructs pseudo-measurements for sensor biases by eliminating the target states. For the constant bias model considered in this paper, the solution of this approach becomes a weighted LS solution as discussed in \cite{lin2005multisensor,Lin2006Multisensor}. For simplicity, we refer to this approach as the pseudo-measurement (PM) approach. The fourth approach is the expectation-maximization (EM) approach developed in \cite{Fortunati2013On} for the relative sensor registration problem, which can be directly applied to solve the asynchronous multi-sensor model considered in this paper. The last approach is the augmented state Kalman filter (ASKF) approach proposed in \cite{zhou2004asy}, which treats the sensor biases as augmented states and simultaneously estimates the target state and sensor biases by a Kalman filter. In our numerical simulations, we implement these approaches on a laptop (Intel Core i5) with Matlab2016 software.
		
		We compare our proposed approach with the aforementioned five approaches in a scenario with $3$ sensors and a moving target. Suppose that the target follows a nearly-constant-velocity model and its initial position and velocity obey $\bm{\xi}_1\sim\mathcal{N}(\bm{\bar{\xi}},10q\mathbf{I}_2)$ and $\bm{\dot{\xi}}_1\sim\mathcal{N}(\mathbf{\bar{v}},q\mathbf{I}_2)$, where $\bm{\bar{\xi}}=[-10,0]^T$ km and $\mathbf{\bar{v}}=[200,0]^T$ m/s, $q$ is the density of the target motion process noise.
		In the scenario, each sensor observes the target every 5 seconds with different starting time. The observation lasts 98 seconds in total and each sensor has 20 measurements. The parameters of the simulation setup are listed in Table \ref{tab:sensors} and the simulation scenario is illustrated in Fig. \ref{fig:simul_single_sc}. We use the root mean square error (RMSE) as the estimation performance metric and use the hybrid Cramer-Rao lower bound (HCRLB) \cite{Fortunati2011Least} as the benchmark. The RMSE and HCRLB are averaged over 500 independent Monte Carlo runs.  In our numerical simulations, {we solve SDP \eqref{eq:SDP} by CVX \cite{Grant2014CVX}. It is observed that the obtained solution of SDP \eqref{eq:SDP} is always of rank one.}
		\begin{table}[h]
			\caption{Simulation setup for the 3-sensor case.} 
			\centering 
			\footnotesize
			\begin{tabular}{ c c c c c c}
				\toprule
				&\textbf{Position} &\textbf{Starting Time} &{{$\Delta \bar{\rho}_m$}}&{{$\Delta \bar{\phi}_m$}}\\
				\midrule
				\textbf{Sensor 1} & $[-5,-10]^T$ km& $0$ s& $-0.8$ km& $2^{\circ}$\\
				\textbf{Sensor 2} & $[5,-10]^T$ km& $1.5$ s& $0.6$ km& $-3^{\circ}$\\
				\textbf{Sensor 3} & $[0,10]^T$ km& $3$ s& $0.8$ km& $-2^{\circ}$\\
				\bottomrule
			\end{tabular}
			
			\label{tab:sensors}
		\end{table}
		
		\begin{figure}[h]
			\centering
			\includegraphics [width=0.6 \linewidth]{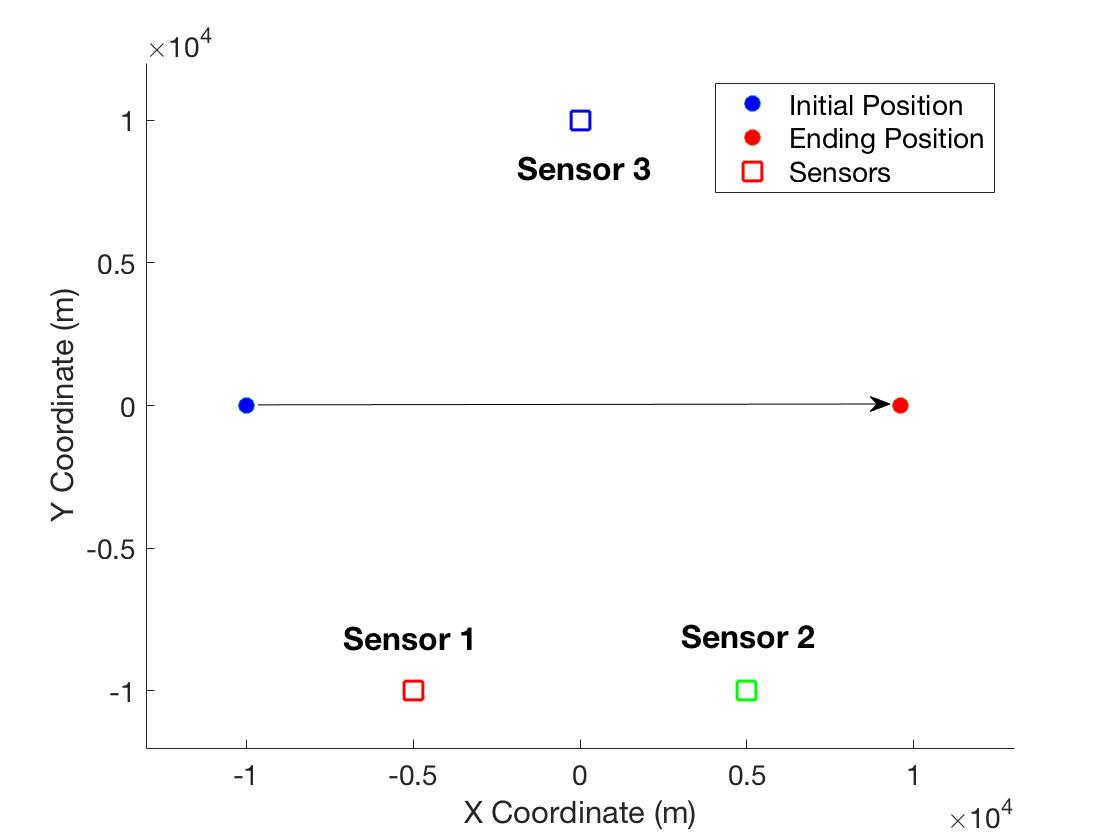}
			\caption{Simulation scenario of three sensors and one target.}
			\label{fig:simul_single_sc}
		\end{figure}
		
		\subsection{Effect of Measurement Noise}\label{subsec:simul_meanoise}
		In this subsection, we compare the performance of the six approaches under different levels of the measurement noise with $q=0.05\ \text{m}^2/\text{s}^3$.
		The RMSEs of the six approaches for the three sensors' range and azimuth biases and the corresponding HCRLBs are plotted as Figs. \ref{fig:simul_multi_s1r1_r}, \ref{fig:simul_multi_s2r1_r}, and \ref{fig:simul_multi_s3r1_r}, where those RMSEs which exceed the range of the Y-axis are plotted by dash lines and the corresponding values are marked. 
		
		We can observe from these figures that the proposed approach always achieves the smallest RMSE among all approaches. The two-stage approach performs the worst and its RMSE increases rapidly as the noise level increases. This is because this approach estimates the range biases only based on each sensor's local measurements and only iterates once. The high level of noise dramatically degrades the range bias estimation and the bad range bias estimation further degrades the estimation accuracy of the azimuth biases.
		\begin{figure}[H]
			\centering
			\includegraphics [width=0.48 \linewidth]{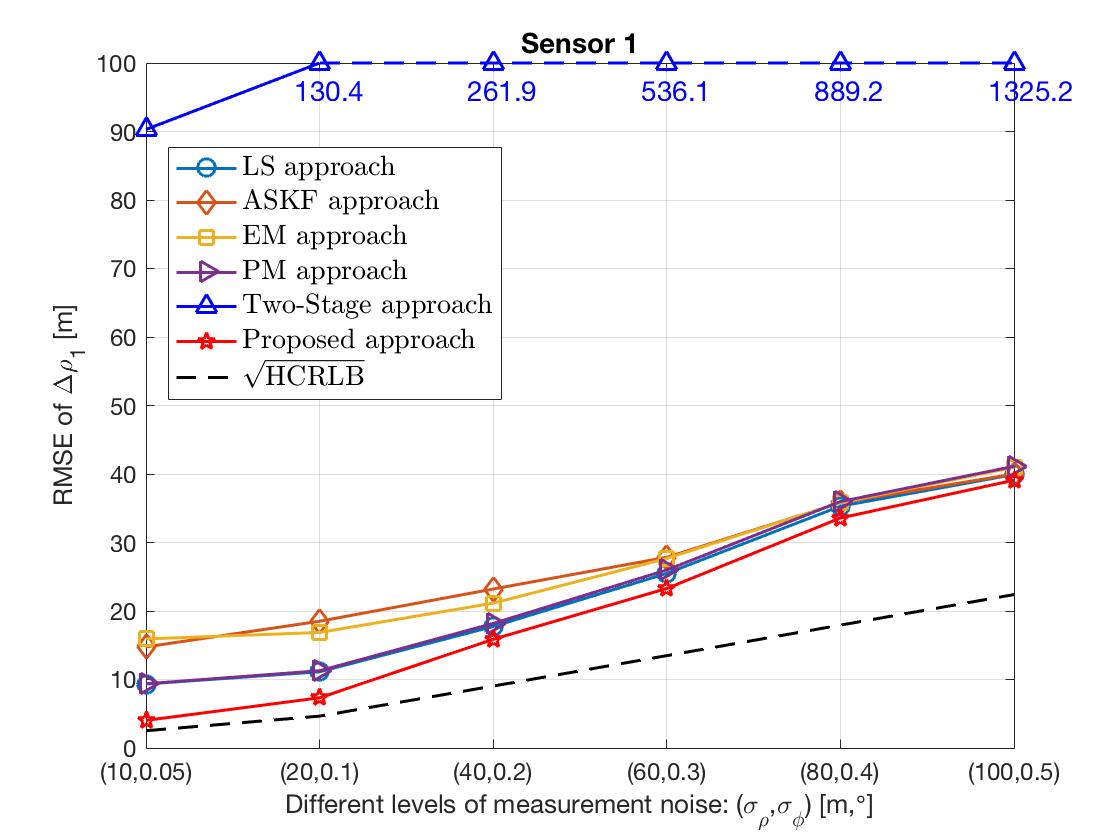}
			\includegraphics [width=0.48 \linewidth]{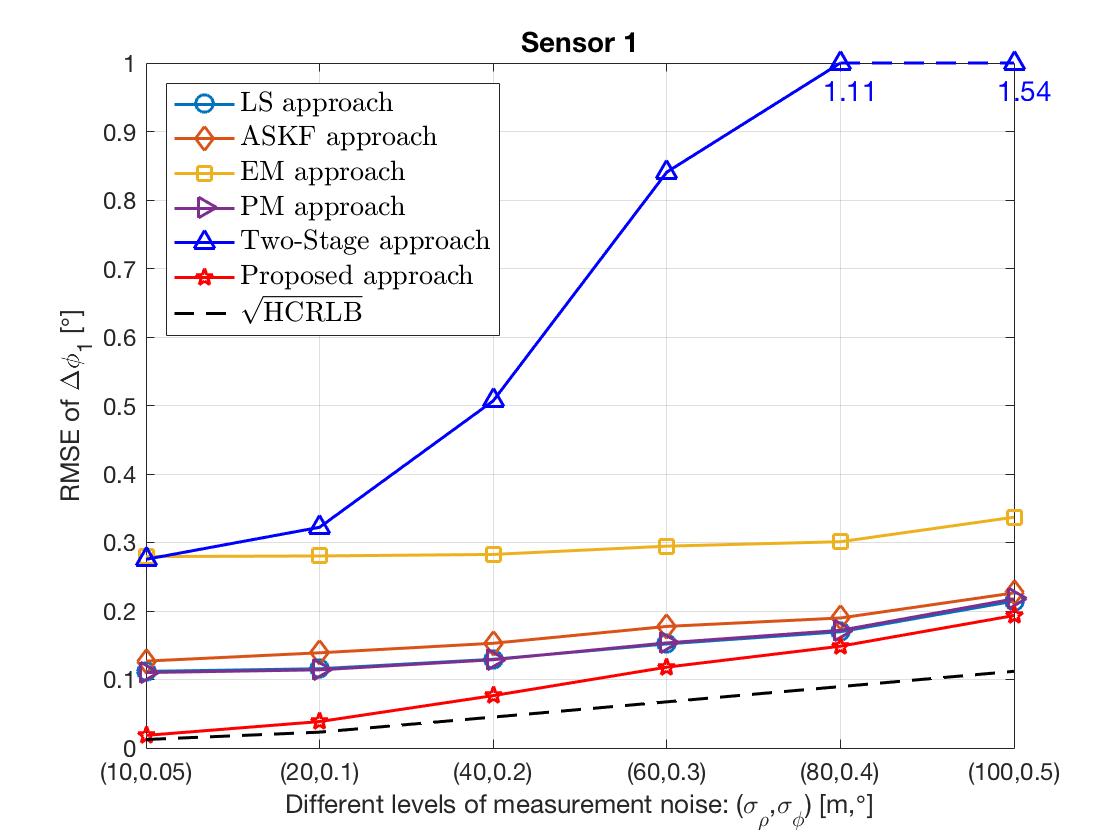}
			\caption{ RMSE and HCRLB of sensor 1's range and azimuth biases under different $\sigma_\rho$ and $\sigma_\phi$.}
			\label{fig:simul_multi_s1r1_r}
		\end{figure}
		
		\begin{figure}[H]
			\centering
			\includegraphics [width=0.48 \linewidth]{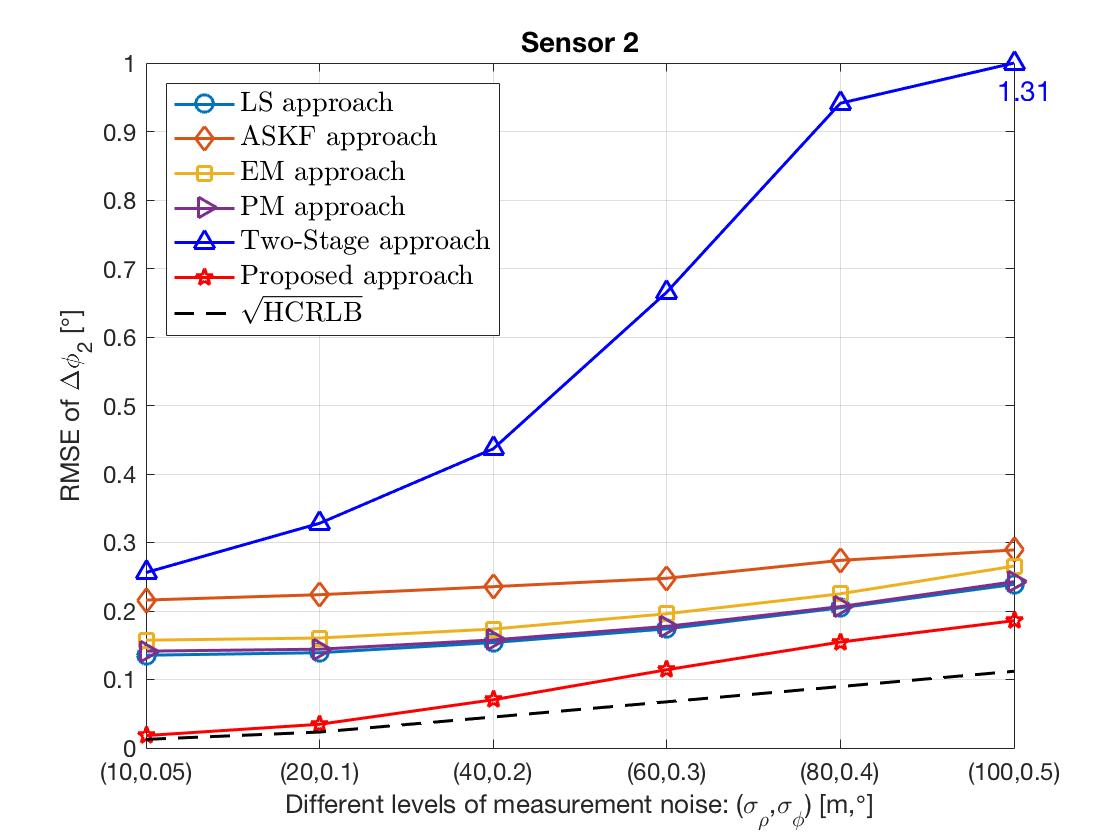}
			\includegraphics [width=0.48 \linewidth]{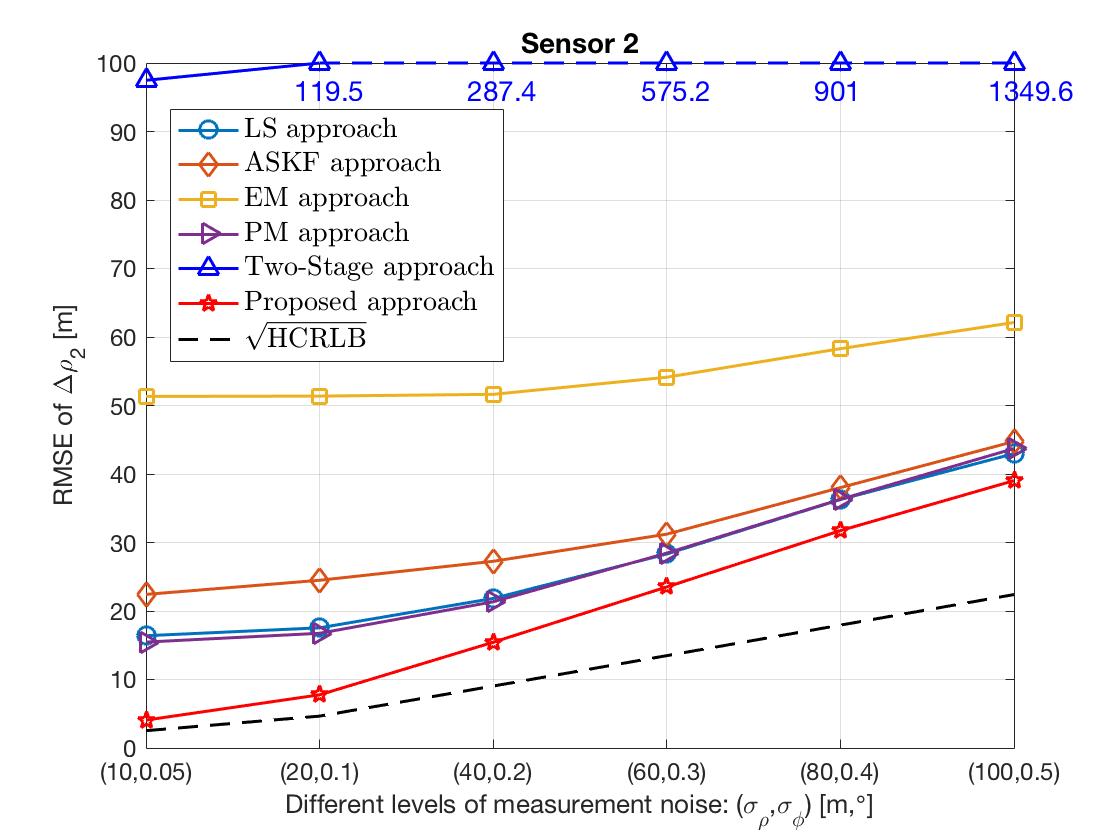}
			\caption{ RMSE and HCRLB of sensor 2's range and azimuth biases under different $\sigma_\rho$ and $\sigma_\phi$.}
			\label{fig:simul_multi_s2r1_r}
		\end{figure}
		
		\begin{figure}[H]
			\centering
			\includegraphics [width=0.48 \linewidth]{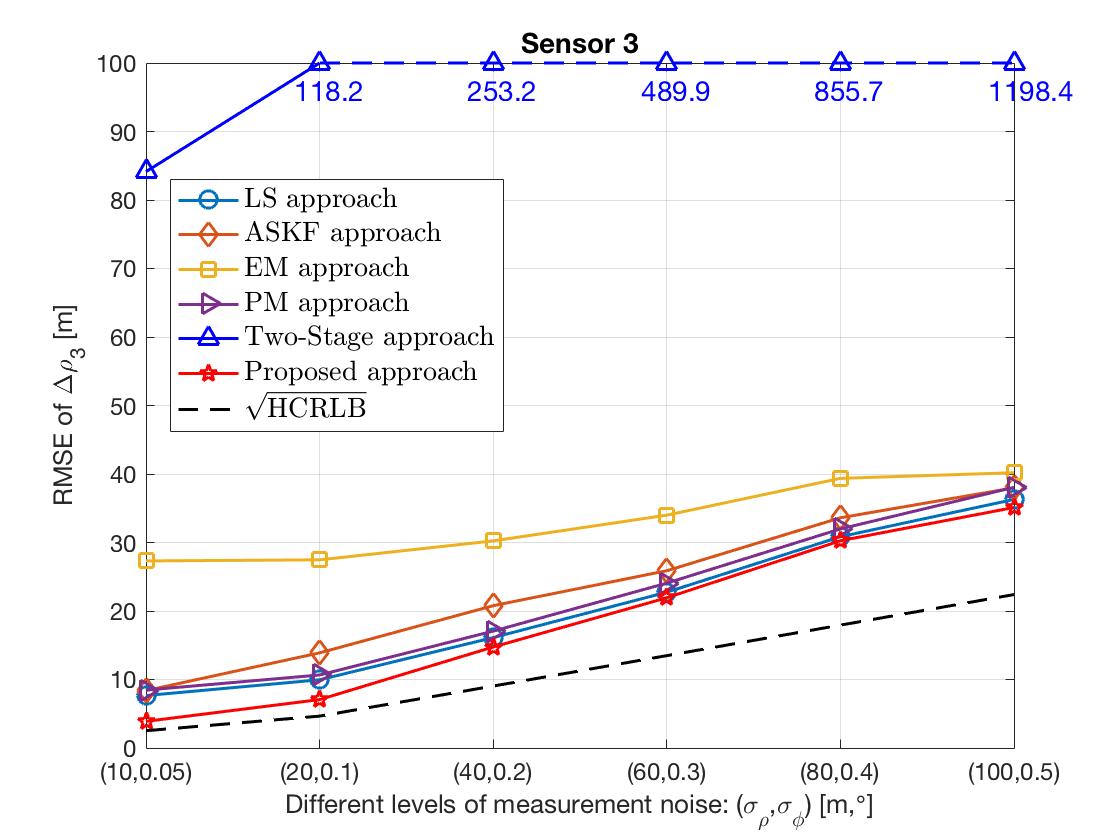}
			\includegraphics [width=0.48 \linewidth]{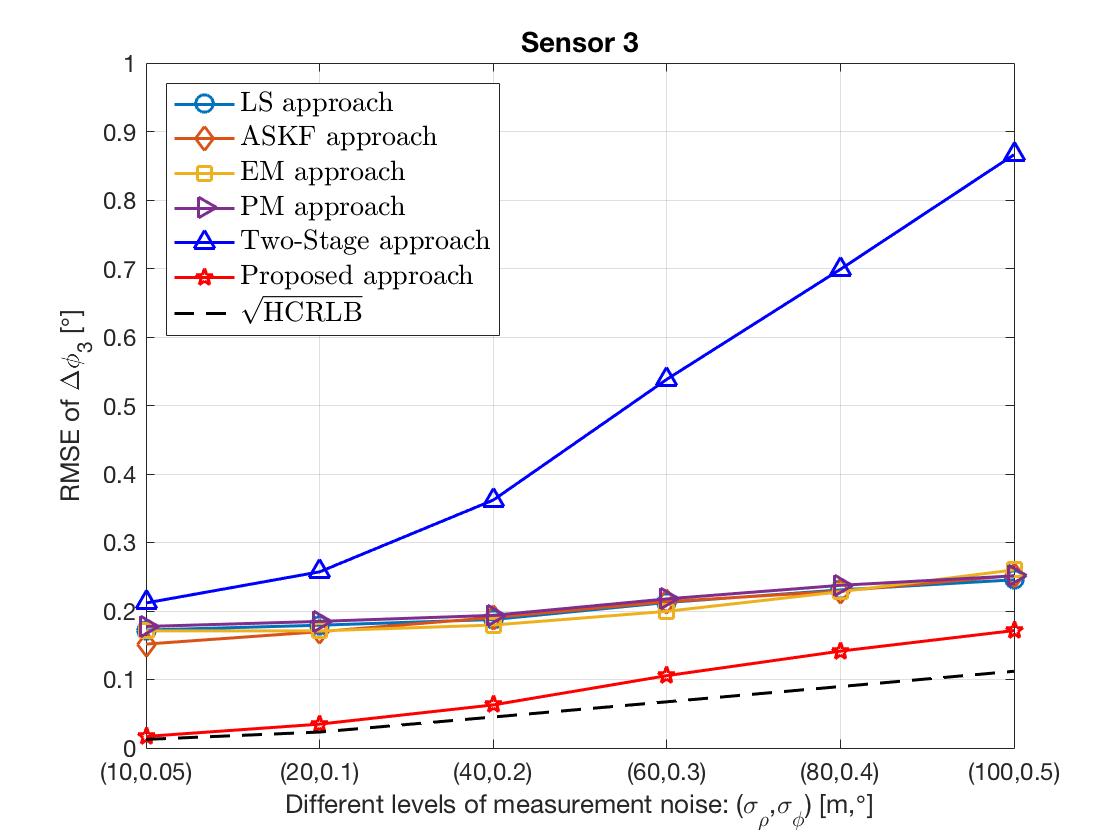}
			\caption{ RMSE and HCRLB of sensor 3's range and azimuth biases under different $\sigma_\rho$ and $\sigma_\phi$.}
			\label{fig:simul_multi_s3r1_r}
		\end{figure}
		
		The linearized LS and PM approaches have a similar performance because both of them are LS type of approaches based on the same first-order approximation procedure that transforms polar measurements into Cartesian measurements. As for the ASKF and EM approaches, it can be found that these two approaches perform similarly but are worse than the linearized LS and PM approaches. This is because both of these two approaches utilize the Kalman filter to estimate the target's states based on an approximated linear measurement model \cite{zhou2004asy}. As the Kalman filter is not robust to the model mismatch, the mismatch error in the approximated linear measurement model degrades the estimation performance of the ASKF and EM approaches even more than the one of the linearized LS and PM approaches.		
		
		Finally, due to the model mismatch error introduced by the first-order approximation procedure, the linearized LS, ASKF, EM, and PM approaches all suffer a `threshold' effect (even) when the measurement noise is very small (i.e., $\sigma_\rho^2=10 \text{ m},\ \sigma_\phi^2=0.05^\circ$). In sharp contrast, the proposed approach exploits the special problem structure and utilizes the nonlinear optimization techniques to deal with the nonlinearity issue, it thus does not have the `threshold' and its RMSE is quite close to the HCRLB when the noise is small.
		
		\subsection{Effect of the Target Motion Process Noise}\label{subsec:simul_qnoise}
		In {this} subsection, we present some simulation results to illustrate the effect of the target motion process noise on the estimation performance. The measurement noise level is fixed at $\sigma_\rho=20$ m and $\sigma_\phi=0.1^{\circ}$ and the process noise density $q$ changes in the interval $[0,10]\ \text{m}^2/\text{s}^3$. The RMSEs of the six approaches and the corresponding HCRLBs are plotted in Figs. \ref{fig:simul_multi_s1_r}, \ref{fig:simul_multi_s2_r}, and \ref{fig:simul_multi_s3_r}.
		\begin{figure}[H]
			\centering
			\includegraphics [width=0.48\linewidth]{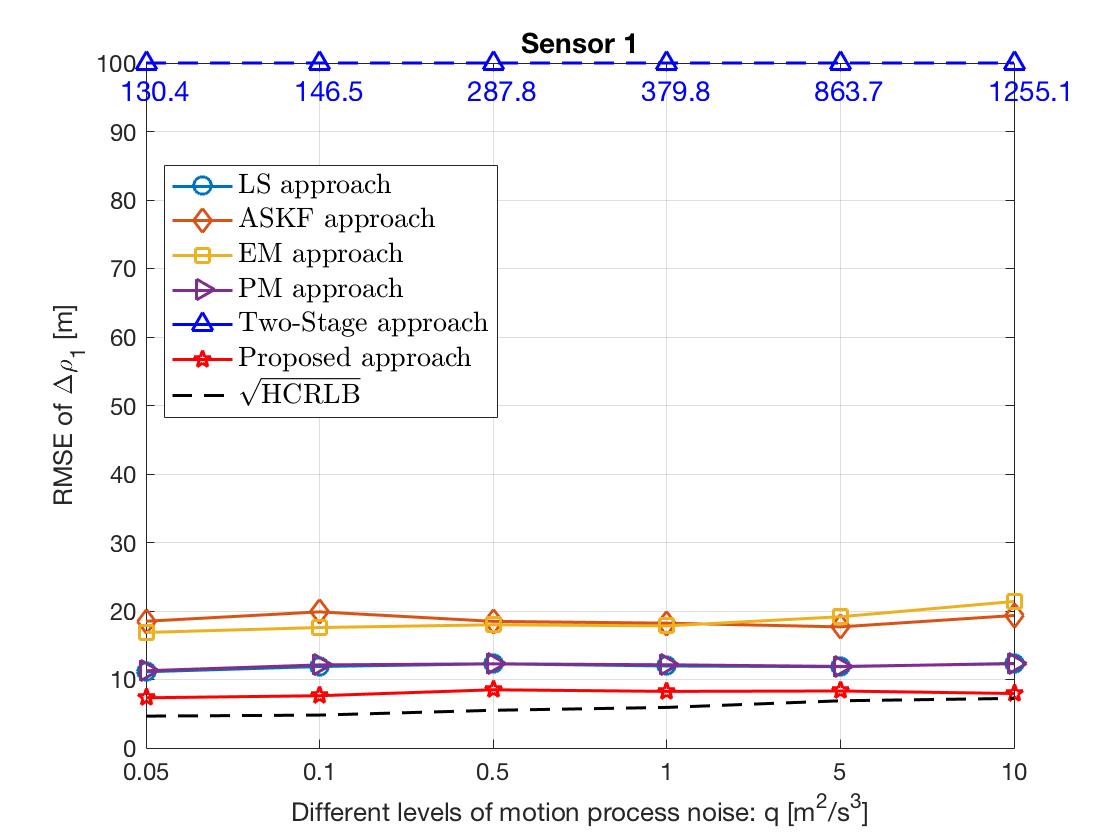}
			\includegraphics [width=0.48\linewidth]{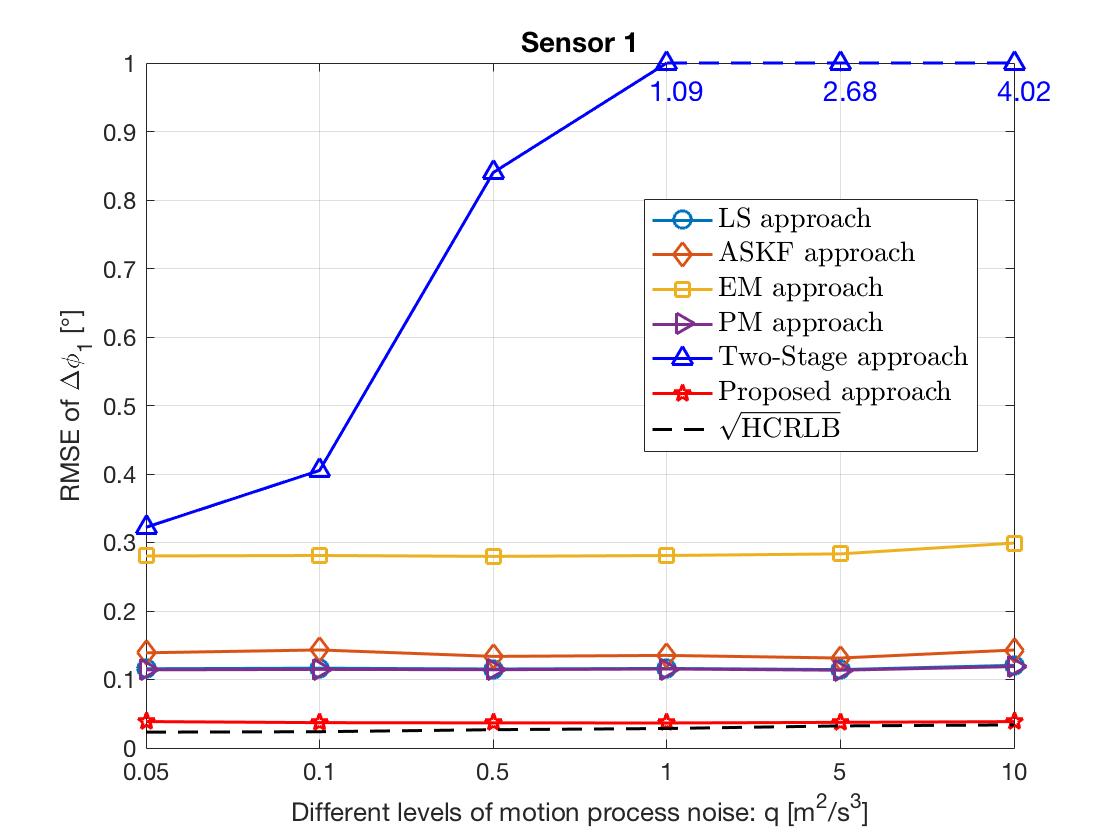}
			\caption{ RMSE and HCRLB of sensor 1's range and azimuth biases under different $q$.}
			\label{fig:simul_multi_s1_r}
		\end{figure}
		
		\begin{figure}[H]
			\centering
			\includegraphics [width=0.48\linewidth]{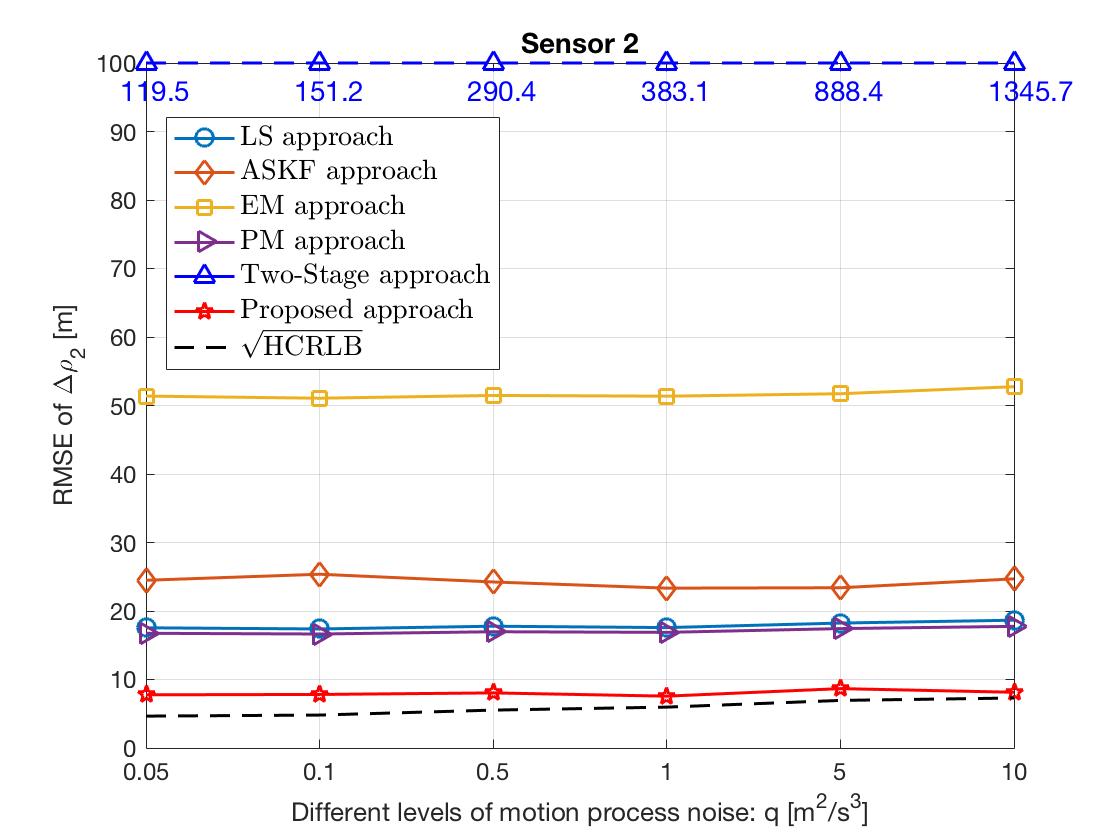}
			\includegraphics [width=0.48\linewidth]{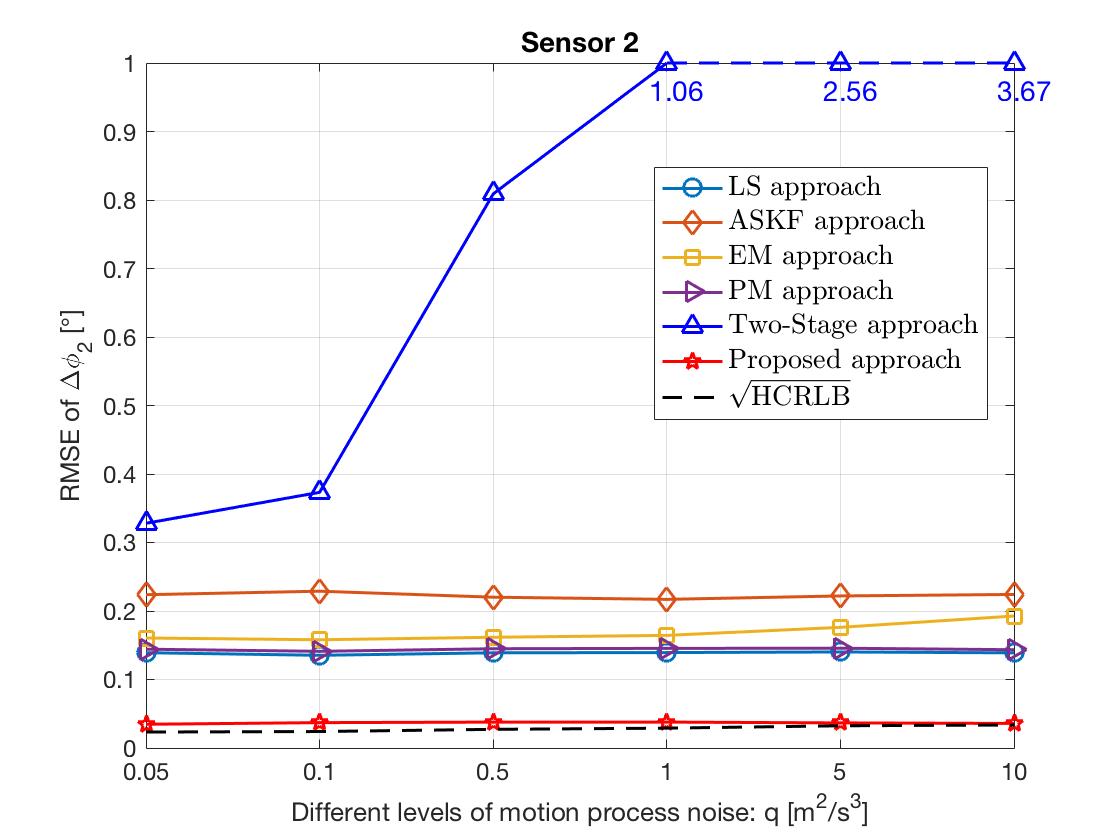}
			\caption{ RMSE and HCRLB of sensor 2's range and azimuth biases under different $q$.}
			\label{fig:simul_multi_s2_r}
		\end{figure}
		
		\begin{figure}[H]
			\centering
			\includegraphics [width=0.48\linewidth]{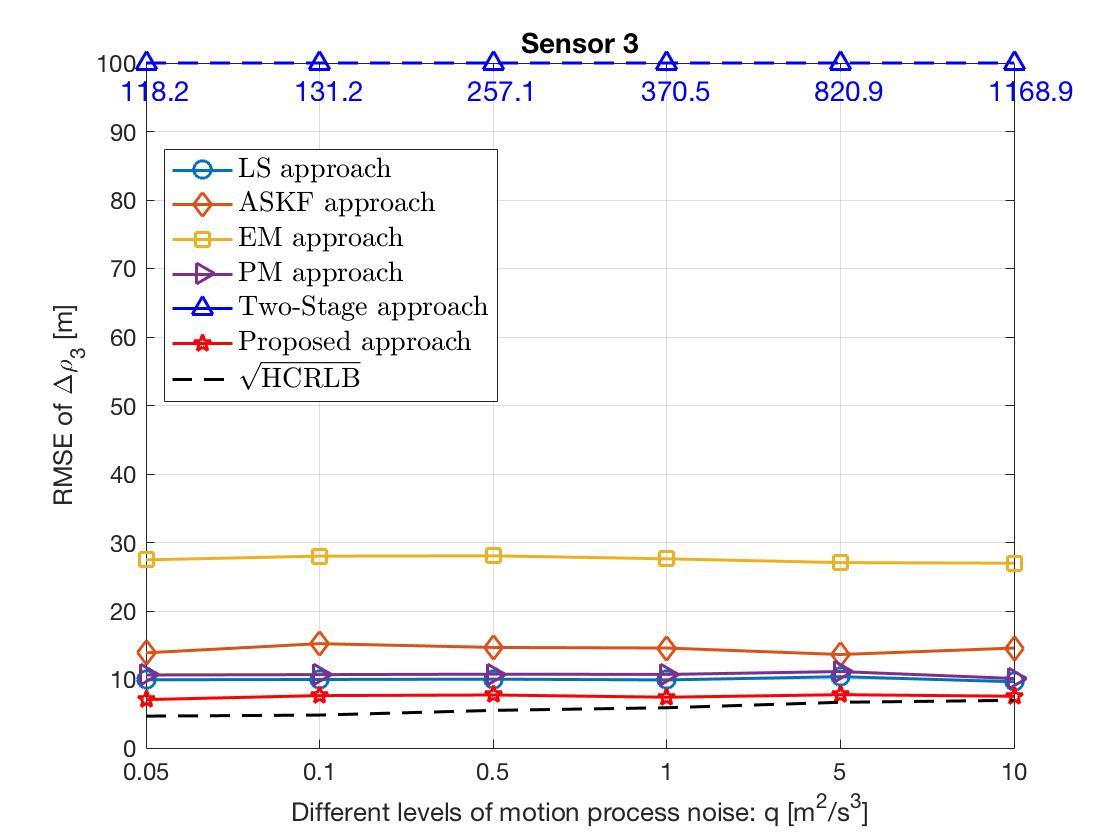}
			\includegraphics [width=0.48\linewidth]{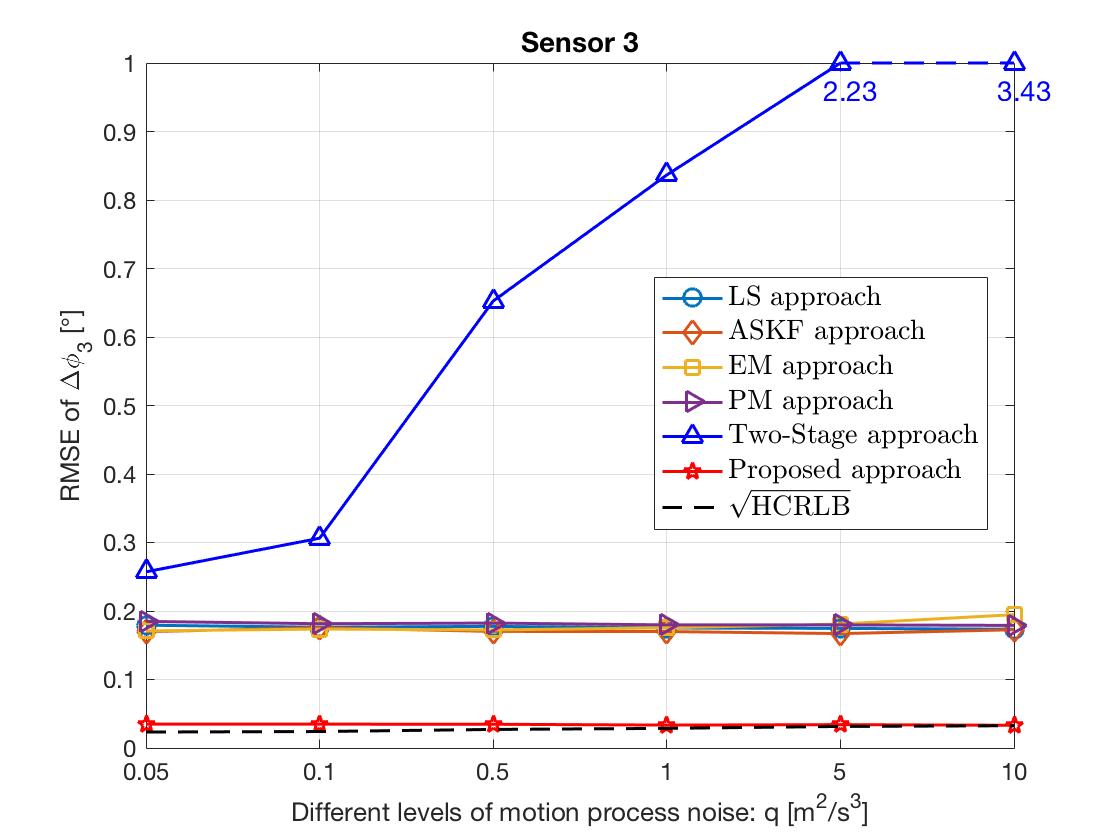}
			\caption{ RMSE and HCRLB of sensor 3's range and azimuth biases under different $q$.}
			\label{fig:simul_multi_s3_r}
		\end{figure}
		
		We can make similar observations from Figs. \ref{fig:simul_multi_s1_r}, \ref{fig:simul_multi_s2_r}, \ref{fig:simul_multi_s3_r} as those from Figs. \ref{fig:simul_multi_s1r1_r}, \ref{fig:simul_multi_s2r1_r}, \ref{fig:simul_multi_s3r1_r}: (i) the two-stage approach is sensitive to the target motion process noise (since it utilizes the sensor's local measurements to estimate the range bias and only iterates once); (ii) compared to the two-stage approach, all other five approaches are somehow robust to the motion process noise (possibly because they combine all sensors' measurements to estimate the range bias and iterate multiple times); (iii) the proposed approach performs better than the other five approaches; (iv) the linearized LS and PM approaches perform similarly and compare favorably against the ASKF and EM approaches.
}}

\subsection{Computational Efficiency Comparison}\label{subsec:simul_time}
{
In this subsection, we compare the computational efficiency of the five approaches (including the LS, EM, ASKF, PM, and proposed approaches) with different numbers of sensors. Specifically, we compare the computational time of these approaches with different numbers of sensors, $M=3,6,12,24.$ The number of measurements of each sensor is fixed to be $10$ and the observation starting time of each sensor is randomly generated in the first $10$ seconds. Each sensor observes the target every $10$ seconds and the locations of those sensors are randomly generated within the area $[-50, 50]$ km $\times$ $[-50, 50]$ km. The range and azimuth biases of  each sensor are also randomly generated within $[-1.5, 1.5]$ km and $[-5^\circ, 5^\circ],$ respectively. The setups for the target is the same as in Subsections \ref{subsec:simul_meanoise} and \ref{subsec:simul_qnoise}. The measurement noise level is set to be $\sigma_\rho=20$ m, $\sigma_\phi=1^{\circ}$, and the noise density of the target motion process is set to be $q=0.05$ $\textrm{m}^2/\textrm{s}^3$.
}

{
In addition to solving QCQP \eqref{EQ_OptDropv} via solving SDP \eqref{eq:SDP}, a more computationally efficient way is to use the gradient projection (GP) method to solve \eqref{EQ_OptDropv}. Although the theoretical convergence of the GP method for solving QCQP \eqref{EQ_OptDropv} is not known, we found that the method always finds the same solution as the solution of SDP \eqref{eq:SDP} in our simulation tests. The computational times of different approaches are listed in Table \ref{tab:time}, where BCD-SDP is the proposed Algorithm 1 and BCD-GP is Algorithm 1 where the GP method is used to solve QCQP \eqref{EQ_OptDropv}.
}
\begin{table}[H]
	\centering
	\caption{Computational Time [seconds].}
	\begin{tabular}{ c c c c c c c}
		\toprule
		&\textbf{LS} &\textbf{EM} &  \textbf{ASKF}&  \textbf{PM}&  \textbf{BCD-SDP} & \textbf{BCD-GP}\\
		\midrule
		$M=3$& 0.002& 0.122& 0.006& 0.004 &0.623& 0.041 \\
		$M=6$& 0.005& 0.335& 0.012& 0.011 & 1.893 & 0.057\\
		$M=12$& 0.007& 0.853& 0.042& 0.025& 3.275&0.083\\
		$M=24$& 0.011& 1.135& 0.143& 0.045& 6.532& 0.325\\
		\bottomrule
	\end{tabular}
	\label{tab:time}
\end{table}
{
It can be observed from Table \ref{tab:time} that the proposed BCD-SDP approach requires more computational time than other approaches as it needs to solve a SDP per iteration. Compared to the BCD-SDP approach, the BCD-GP approach returns the same solution but with significantly less computational time in all cases. Moreover, the computational times of the LS and PM approaches are smaller than all the other iterative approaches (EM, ASKF, BCD-SDP and BCD-GP) in all cases, because the optimization problems in these two approaches have closed-form solutions.
}

\section{Concluding Remarks}

{\textcolor{black}{In this paper, we proposed an effective BCD algorithm for the asynchronous multi-sensor registration problem by exploiting the presence of a reference target moving {with} a \textcolor{black}{nearly} constant velocity. Our proposed algorithm effectively handles the range and azimuth biases by exploiting the special structure of the problem. In particular, we revealed an interesting separation property of the range bias estimation problem, which shows that each sensor can actually estimate its range bias from its local measurements. We then applied the SDR technique to solve the azimuth bias estimation problem and further showed that the corresponding SDR is tight (under mild conditions, see Theorem \ref{thm:multi}). In sharp contrast to the existing existing approaches, our proposed algorithm is capable of exactly recovering sensor biases when there is no noise. The simulation results show the effectiveness of our proposed algorithm in various noisy scenarios. 
}}

{
As shown in Subsection \ref{subsec:simul_time}, the GP method is a more computationally efficient algorithm for solving QCQP \eqref{EQ_OptDropv}.
Although there is no theoretical guarantee, the GP algorithm always finds the same solution as the SDP solution in our simulation tests. Some recent theoretical progress on using the GP algorithm to solve a similar non-convex problem has been reported in \cite{liu2017discrete}. It would be interesting to also obtain some convergence and optimality guarantee of the GP algorithm for solving QCQP \eqref{EQ_OptDropv}. Moreover, it will be interesting to generalize the results in the paper from the 2-dimensional case to the 3-dimensional case. The results developed in the 2-dimensional case should provide good insight to solve the problem in the 3-dimensional case, although more efforts are still needed to achieve high estimation accuracy in view of the higher nonlinearity as well as more involved noise \cite{Fortunati2011Least,Fortunati2013On} in the 3-dimensional case. Finally, we plan to study the statistical property of the optimization model in two steps as we outline below. First, recall we have already shown that if the noise level is sufficiently low, then the optimization model will be able to give a globally optimal solution. To understand the statistical quality of such an optimal solution, we need to relate it to the maximum likelihood estimate (MLE) of the registration problem, and then bound the distance of MLE to the true biases. The latter problem is classical for a maximum likelihood estimation problem (related to Cramer-Rao lower bound). Second, for a given noise level, more data samples effectively reduce the noise level in the system. We need to estimate how many data samples are needed to bring the equivalent noise level to be under the `threshold' established in Theorem \ref{thm:multi}. Combining the above two steps will yield a statistical statement of the form: if the number of data samples are above a certain bound, then the optimization model would give a rank-one solution which is within a certain distance from the true bias parameters.
}

\section*{Appendix A: Proof of Theorem \ref{thm:single}}\label{appendix-theorem1}	
	We first reformulate problem \eqref{eq:single_fixphi} as
	\begin{equation}\label{eq:thm1_H0H1}
	\min_{\Delta\rho}\min_{\mathbf{v}}~~\left\|\mathbf{H}_{0}\Delta\rho+\mathbf{H}_{1}\mathbf{v}-\mathbf{y} \right\|^2.
	\end{equation} The closed-form solution of the above problem with respect to $\mathbf{v}$ is 
	\begin{equation} \label{eq:thm1_v}
	\mathbf{v}^*=-\left(\mathbf{H}_1^T\mathbf{H}_1\right)^{-1}\mathbf{H}_1^T\left(\mathbf{H}_0\Delta\rho-\mathbf{y}\right).
	\end{equation}
	Substituting \eqref{eq:thm1_v} into \eqref{eq:thm1_H0H1}, we obtain the following linear {LS} problem with respect to $\Delta\rho:$
	\begin{equation}\label{eq:thm1_rho}
	\min_{\Delta\rho}~~\left\| \left(\mathbf{I}_{2(K-1)}-\mathbf{\bar{H}}_1\right)\left(\mathbf{H}_0\Delta\rho-\mathbf{y}\right)\right\|^2,
	\end{equation}
	where $\mathbf{\bar{H}}_1=\mathbf{H}_1\left(\mathbf{H}_1^T\mathbf{H}_1\right)^{-1}\mathbf{H}_1^T$. By the definitions of $\mathbf{H}_0,~\mathbf{H}_1,$ and $\mathbf{y}$ in \eqref{eq:single_Handy}, problem \eqref{eq:thm1_rho} can be further rewritten as
	\begin{equation*}
	\min_{\Delta\rho}~~\sum_{k=1}^{K-1}\displaystyle \left\|
	\mathbf{a}_k\Delta\rho - \mathbf{b}_k
	\right\|^2,
	\end{equation*}
	where
	\begin{align*}
	\mathbf{a}_k=\left[\displaystyle\sum_{\ell=1}^{K-1}\bar{T}_\ell^kc_\ell,~\displaystyle \sum_{\ell=1}^{K-1}\bar{T}_\ell^ks_\ell\right]^T,\quad
	\mathbf{b}_k=\left[\displaystyle\sum_{\ell=1}^{K-1}\bar{T}_\ell^ky_\ell^c,~\displaystyle \sum_{\ell=1}^{K-1}\bar{T}_\ell^ky_\ell^s\right]^T,
	\end{align*}
	and
	\begin{equation*}
	\bar{T}_\ell^k=\left\lbrace \begin{aligned}&1-\frac{T_k^2}{\sum_{i=1}^{K-1}T_i^2},\ &\ell=k,\\[5pt]
	&-\frac{T_kT_\ell}{\sum_{i=1}^{K-1}T_i^2},\ &\ell\neq k.\end{aligned}\right.
	\end{equation*}
	
	To prove Theorem \ref{thm:single}, it suffices to show that $\mathbf{a}_k^T\mathbf{a}_k$ and $\mathbf{a}_k^T\mathbf{b}_k$ for all $k=1,2,\ldots,K-1$ are independent of $\Delta\phi.$ Next, we show that this is indeed true.
	It is simple to compute, for $k=1,2,\ldots,K,$
	\begin{align*}
	\mathbf{a}_k^T\mathbf{a}_k
	=\sum_{\ell=1}^{K-1}\sum_{j=1}^{K-1}\bar{T}_\ell^k\bar{T}_j^k\left(c_\ell c_j+s_\ell s_j\right),\quad 
	\mathbf{a}_k^T\mathbf{b}_k
	=\sum_{\ell=1}^{K-1}\sum_{j=1}^{K-1}\bar{T}_\ell^k\bar{T}_j^k\left(c_\ell y^c_j+s_\ell y^s_j\right).
	\end{align*}
	Moreover, by the definitions of $c_k, s_k, y_k^c,$ and $y_k^s$ in \eqref{eq:c_k}, \eqref{eq:s_k}, \eqref{eq:yc_k}, and \eqref{eq:ys_k} and the fact
	\begin{align*}
	\cos(\alpha_1+\Delta\phi)\cos(\alpha_2+\Delta\phi)+\sin(\alpha_1+\Delta\phi)\sin(\alpha_2+\Delta\phi)
	=\cos(\alpha_1-\alpha_2),
	\end{align*}
	we can obtain, for $\ell, j=1,2,\ldots,K-1,$
	\begin{align*}c_\ell c_j+s_\ell s_j=&\lambda^{-2}[\cos\left(\phi_{\ell+1}-\phi_{j+1}\right)+\cos\left(\phi_{\ell}-\phi_{j}\right)\\
	&-\cos\left(\phi_{\ell+1}-\phi_{j}\right)-\cos\left(\phi_{j+1}-\phi_{\ell}\right)],\end{align*}
	\begin{align*}c_\ell y_j^c+s_\ell y_j^s=&\lambda^{-2}[\rho_{j+1}\cos\left(\phi_{\ell+1}-\phi_{j+1}\right)+\rho_{j}\cos\left(\phi_{\ell}-\phi_{j}\right)\\
	&-\rho_{j}\cos\left(\phi_{\ell+1}-\phi_{j}\right)-{{\rho_{j+1}}}\cos\left(\phi_{j+1}-\phi_{\ell}\right)].\end{align*}
	This immediately shows that $\mathbf{a}_k^T\mathbf{a}_k$ and $\mathbf{a}_k^T\mathbf{b}_k$ for all $k=1,2,\ldots,K-1$ are independent of $\Delta\phi.$ This completes the proof of Theorem \ref{thm:single}. \qed

\section*{Appendix B: Proof of Theorem \ref{thm:multi}}\label{app:thm2}
To show that SDP \eqref{eq:SDP} admits a unique solution of rank one, it is sufficient to show that, if $\mathbf{\Delta H}$ is sufficiently small, there always exists $\mathbf{W}=\textrm{Diag}(\mathbf{w})$ with $|\mathbf{w}_m|=1$ for all $m=1,2,\ldots,M$ such that $\mathbf{1}\mathbf{1}^T\in\mathbb{H}^{M+1}$ is the unique solution of the following SDP
\begin{equation}
\begin{aligned}\label{eq:SDP2}
\min_{\mathbf{X}\in\mathbb{H}^{M+1}}\quad&\textrm{Tr}(\mathbf{\hat{C}X})\\
\textrm{s.t.}\quad\quad&\textrm{diag}(\mathbf{X})=\mathbf{1}, \\
&\mathbf{X}\succeq\mathbf{0},
\end{aligned}
\end{equation}%
where
\begin{equation*}
\mathbf{\hat{C}}=\begin{bmatrix}\mathbf{W}^\dagger\mathbf{H}^\dagger\mathbf{PHW}&\mathbf{W}^\dagger\mathbf{H}^\dagger\mathbf{P}\mathbf{c} \\ \mathbf{c}^\dagger\mathbf{PHW}&0\end{bmatrix}\in\mathbb{H}^{M+1}.
\end{equation*}
By Lemma \ref{lem:sufcond}, we only need to show that, if $\mathbf{\Delta H}$ is sufficiently small, there always exists $\mathbf{W}=\textrm{Diag}(\mathbf{w})$ with $|\mathbf{w}_m|=1$ for all $m=1,2,\ldots,M$ such that $\mathbf{1}\mathbf{1}^T\in\mathbb{H}^{M+1}$ and some $\mathbf{y}\in\mathbb{R}^{M+1}$ jointly satisfy
\begin{equation}\label{semidefinite}
\begin{aligned}
&\mathbf{\hat{C}}+\textrm{diag}(\mathbf{y})\\
&=\begin{bmatrix}\mathbf{W}^\dagger\mathbf{H}^\dagger\mathbf{PHW}+\textrm{Diag}(\mathbf{y}_{1:M})&\mathbf{W}^\dagger\mathbf{H}^\dagger\mathbf{P}\mathbf{c} \\ \mathbf{c}^\dagger\mathbf{PHW}&\mathbf{y}_{M+1}\end{bmatrix}\succeq \mathbf{0},
\end{aligned}
\end{equation}

\begin{equation}\label{eq:lem1cond3}
\begin{bmatrix}\mathbf{W}^\dagger\mathbf{H}^\dagger\mathbf{PHW}+\textrm{Diag}(\mathbf{y}_{1:M})&\mathbf{W}^\dagger\mathbf{H}^\dagger\mathbf{P}\mathbf{c} \\ \mathbf{c}^\dagger\mathbf{PHW}&\mathbf{y}_{M+1}\end{bmatrix}\mathbf{1}=\mathbf{0}\in\mathbb{R}^{M+1}, 
\end{equation} and
\begin{equation}\label{strictpositive}
\mathbf{W}^\dagger\mathbf{H}^\dagger\mathbf{PHW}+\textrm{Diag}(\mathbf{y}_{1:M}) \succ \mathbf{0}.
\end{equation}
Notice that the above \eqref{semidefinite}, \eqref{eq:lem1cond3}, and \eqref{strictpositive} correspond to conditions 2, 3, and 4 in Lemma \ref{lem:sufcond}, respectively, and $\mathbf{1}\mathbf{1}^T\in\mathbb{S}^{M+1}$ obviously satisfies condition 1 in Lemma \ref{lem:sufcond}. Moreover, we rewrite \eqref{eq:lem1cond3} as follows:
\begin{subequations}
	\begin{align}\label{eq:lem1cond3a}
	\left[\mathbf{W}^\dagger\mathbf{H}^\dagger\mathbf{PHW}+\textrm{Diag}(\mathbf{y}_{1:M}) \right]\mathbf{1}+\mathbf{W}^\dagger\mathbf{H}^\dagger\mathbf{P}\mathbf{c}=\mathbf{0}\in\mathbb{R}^{M},
	\end{align}
	\begin{align}\label{eq:lem1cond3b}
	\left[\mathbf{c}^\dagger\mathbf{PHW}\right]\mathbf{1}+\mathbf{y}_{M+1}=0.
	\end{align}
\end{subequations}

Next, we prove that, if $\mathbf{\Delta H}$ is sufficiently small, there always exist $\mathbf{W}=\textrm{Diag}(\mathbf{w})$ with $|\mathbf{w}_m|=1$ for all $m=1,2,\ldots,M$ and $\mathbf{y}\in\mathbb{R}^{M+1}$that satisfy \eqref{semidefinite}, \eqref{eq:lem1cond3a}, \eqref{eq:lem1cond3b}, and \eqref{strictpositive}. We divide the proof into two parts. More specifically, we first show that, if there exist $\mathbf{W}=\textrm{Diag}(\mathbf{w})$ with $|\mathbf{w}_m|=1$ for all $m=1,2,\ldots,M$ and $\mathbf{y}\in\mathbb{R}^{M+1}$ satisfying \eqref{eq:lem1cond3a} and \eqref{strictpositive}, then we can construct $\mathbf{y}_{M+1}$ that simultaneously satisfies \eqref{semidefinite} and \eqref{eq:lem1cond3b}. Then, we show that, if $\mathbf{\Delta H}$ is sufficiently small, there indeed exist $\mathbf{W}=\textrm{Diag}(\mathbf{w})$ with $|\mathbf{w}_m|=1$ for all $m=1,2,\ldots,M$ and $\mathbf{y}_{1:M}$ that satisfy \eqref{eq:lem1cond3a} and \eqref{strictpositive}.

Let us first assume that $\mathbf{W}=\textrm{Diag}(\mathbf{w})$ (with $|\mathbf{w}_m|=1$ for all $m=1,2,\ldots,M$) and $\mathbf{y}_{1:M}$ satisfy \eqref{strictpositive} and \eqref{eq:lem1cond3a}. We show that there exists $\mathbf{y}_{M+1}$ that simultaneously satisfies \eqref{semidefinite} and \eqref{eq:lem1cond3b}. Let
\begin{equation}\label{constructyM+1}
\mathbf{y}_{M+1}=\mathbf{c}^\dagger\mathbf{PHW}\left[\mathbf{W}^\dagger\mathbf{H}^\dagger\mathbf{PHW}+\textrm{Diag}(\mathbf{y}_{1:M}) \right]^{-1}\mathbf{W}^\dagger\mathbf{H}^\dagger\mathbf{P}\mathbf{c}.
\end{equation}
This, together with \eqref{strictpositive}, immediately implies \eqref{semidefinite}. Moreover, from \eqref{eq:lem1cond3a}, we have $$\mathbf{1}=-\left[\mathbf{W}^\dagger\mathbf{H}^\dagger\mathbf{PHW}+\textrm{Diag}(\mathbf{y}_{1:M}) \right]^{-1}\mathbf{W}^\dagger\mathbf{H}^\dagger\mathbf{P}\mathbf{c}\in\mathbb{R}^{M}.$$
Combining the above and \eqref{constructyM+1} immediately yields \eqref{eq:lem1cond3b}.

Now, let us argue that, if $\mathbf{\Delta H}$ is sufficiently small, then there always exist $\mathbf{W}=\textrm{Diag}(\mathbf{w})$ (with $|\mathbf{w}_m|=1$ for all $m=1,2,\ldots,M$) and $\mathbf{y}_{1:M}$ that satisfy \eqref{eq:lem1cond3a} and \eqref{strictpositive}. In particular, the following Claim 1 and Claim 2 guarantee the existence of $\mathbf{W}=\textrm{Diag}(\mathbf{w})$ (with $|\mathbf{w}_m|=1$ for all $m=1,2,\ldots,M$) and $\mathbf{y}_{1:M}$ that satisfy \eqref{eq:lem1cond3a} and \eqref{strictpositive}, respectively.



\noindent \textbf{Claim 1.} There always exist a neighborhood ${\mathcal{H}}\subseteq\mathbb{C}^{(K-1)\times M}$ containing $\mathbf{\tilde H}$ and two unique continuously differentiable functions $d_{\bm{\psi}}:\mathbb{C}^{(K-1)\times M}\mapsto\mathbb{R}^{M}$ and $d_{\mathbf{y}}:\mathbb{C}^{(K-1)\times M}\mapsto\mathbb{R}^{M}$ such that \eqref{eq:lem1cond3a} holds true for all $\mathbf{H}\in \mathcal{H}$ with $\mathbf{W}=\textrm{Diag}\left(\left[e^{j\psi_1}, e^{j\psi_2},\ldots, e^{j\psi_M}\right]^T\right)$, $\bm{\psi}=d_{\bm{\psi}}(\mathbf{H}),$ and $\mathbf{y}=d_{\mathbf{y}}(\mathbf{H})$.
\begin{proof}
	To prove Claim 1, we first reformulate complex {equation \eqref{eq:lem1cond3a} as an equivalent real form; 
		then we} apply the implicit function theorem \cite{border2013notes} based on the equivalent form to show the existence of the neighborhood $\mathcal{H}$ and two unique continuously differentiable functions $d_{\bm{\psi}}$ and $d_{\mathbf{y}}.$
	%
	%
	
	The equivalent form of \eqref{eq:lem1cond3a} is
	{{
			\begin{equation}\label{fequationf}f\left(\bm{\psi},\mathbf{y},\mathbf{H}\right)=\bm{0},\end{equation}
			where $f\left(\bm{\psi},\mathbf{y},\mathbf{H} \right):\mathbb{R}^M\times\mathbb{R}^M\times\mathbb{C}^{(K-1)\times M}\mapsto\mathbb{R}^{2M}$ is	 \begin{equation}\label{eq:definef}
			\begin{aligned}
			f\left(\bm{\psi},\mathbf{y},\mathbf{H}\right)
			=&\begin{bmatrix}\text{Re}\left\{ 	\left[\mathbf{W}^\dagger\mathbf{H}^\dagger\mathbf{PHW}+\textrm{diag}(\mathbf{y}_{1:M}) \right]\mathbf{1}+\mathbf{W}^\dagger\mathbf{H}^\dagger\mathbf{P}\mathbf{c}\right\}\\[3pt] \text{Im}\left\{ 	 \left[\mathbf{W}^\dagger\mathbf{H}^\dagger\mathbf{PHW}+\textrm{diag}(\mathbf{y}_{1:M}) \right]\mathbf{1}+\mathbf{W}^\dagger\mathbf{H}^\dagger\mathbf{P}\mathbf{c}\right\} \end{bmatrix}\\
			=&\begin{bmatrix}
			\mathbf{C_{\bm{\psi}}\text{Re}\{ A_H \}}\mathbf{C_{\bm{\psi}}}+\mathbf{S_{\bm{\psi}}\text{Im}\{ A_H \}}\mathbf{C_{\bm{\psi}}}\\
			\mathbf{C_{\bm{\psi}}\text{Re}\{ A_H \}}\mathbf{S_{\bm{\psi}}}+\mathbf{S_{\bm{\psi}}\text{Im}\{ A_H \}}\mathbf{S_{\bm{\psi}}}
			\end{bmatrix}\mathbf{1}\\
			 &+ \begin{bmatrix}
			-\mathbf{C_{\bm{\psi}}\text{Im}\{ A_H \}}\mathbf{S_{\bm{\psi}}}+\mathbf{S_{\bm{\psi}}\text{Re}\{ A_H \}}\mathbf{S_{\bm{\psi}}}\\
			\mathbf{C_{\bm{\psi}}\text{Im}\{ A_H \}}\mathbf{C_{\bm{\psi}}}-\mathbf{S_{\bm{\psi}}\text{Re}\{ A_H \}}\mathbf{C_{\bm{\psi}}}
			\end{bmatrix}\mathbf{1}\\
			& +
			\begin{bmatrix}
			\mathbf{C_{\bm{\psi}}}\mathbf{\text{Re}\{ b_H \}}+\mathbf{S_{\bm{\psi}}}\mathbf{\text{Im}\{ b_H \}}\\
			\mathbf{C_{\bm{\psi}}}\mathbf{\text{Im}\{ b_H \}}-\mathbf{S_{\bm{\psi}}}\mathbf{\text{Re}\{ b_H \}}
			\end{bmatrix}
			+\begin{bmatrix}\mathbf{y}\\\bm{0}\end{bmatrix}.
			\end{aligned}
			\end{equation}}}In the above, $\mathbf{C_{\bm{\psi}}}$ and $\mathbf{S_{\bm{\psi}}}$ are diagonal matrices related to $\bm{\psi}$ as follows:
	\begin{equation}\label{eq:defineCS}
	\begin{aligned}
	 \mathbf{C_{\bm{\psi}}}=\textrm{Diag}\left(\left[\cos\bm{\psi}_1,\ldots,\cos\bm{\psi}_M\right]^T\right),\\
	 \mathbf{S_{\bm{\psi}}}=\textrm{Diag}\left(\left[\sin\bm{\psi}_1,\ldots,\sin\bm{\psi}_M\right]^T\right);
	\end{aligned}
	\end{equation}
	and $\mathbf{A_H}\in\mathbb{H}^{M}$ and $\mathbf{ b_H}\in\mathbb{C}^M$ are determined by $\mathbf{H}$ as follows:
	\begin{equation*}
	\mathbf{A_H} = \mathbf{H^\dagger PH},~\mathbf{ b_H} = \mathbf{H^\dagger Pc}.
	\end{equation*}
	
	Now, we apply the implicit function theorem \cite[Theorem 5]{border2013notes} to Eq. \eqref{fequationf}.
	Notice that it is the same to define $f\left(\bm{\psi},\mathbf{y},\mathbf{H}\right)$ in \eqref{eq:definef} over $\left(\bm{\psi},\mathbf{y},\mathbf{H}\right)$ or define $f(\bm{\psi},\mathbf{y},\text{Re}\{\mathbf{{H}}\},\text{Im}\{\mathbf{{H}}\}): \mathbb{R}^M\times\mathbb{R}^M\times\mathbb{R}^{(K-1)\times M}\times\mathbb{R}^{(K-1)\times M}\mapsto\mathbb{R}^{2M}$ over $(\bm{\psi},\mathbf{y},\text{Re}\{\mathbf{{H}}\},\text{Im}\{\mathbf{{H}}\}).$ For {notation}al simplicity, we will keep using \eqref{eq:definef} (but we can think it as $f(\bm{\psi},\mathbf{y},\text{Re}\{\mathbf{{H}}\},\text{Im}\{\mathbf{{H}}\})$) in our proof. First of all, {$f(\bm{\psi},\mathbf{y},\mathbf{{H}})$} is continuously differentiable (i.e., $f(\bm{\psi},\mathbf{y},\text{Re}\{\mathbf{{H}}\},\text{Im}\{\mathbf{{H}}\})$ is continuously differentiable).
	%
	Moreover, it follows from \eqref{eq:NoiseFree} that
	{{
			\begin{equation}\label{eq:define_zeropoint}
			f\left(\bm{\tilde{\psi}},\mathbf{\tilde{y}},\mathbf{\tilde{H}}\right)=\mathbf{0},
			\end{equation}}}where $\bm{\tilde{\psi}}=\bm{\Delta\bar{\phi}},~\mathbf{\tilde{y}}=\mathbf{0},$ and $\mathbf{\tilde{H}}$ is defined in \eqref{eq:tilde_H}.
	If the Jacobian matrix $D_{(\bm{\psi},\mathbf{y})}f$ is invertible at point $(\bm{\tilde{\psi}},\mathbf{\tilde{y}},\mathbf{\tilde{H}})$, then it follows from the implicit function theorem \cite[Theorem 5]{border2013notes} that the desired $\mathcal{{H}}$, ${d}_{\bm{\psi}},$ and ${d}_{\mathbf{y}}$ in Claim 1 must exist.
	It remains to prove that the Jacobian matrix $D_{(\bm{\psi},\mathbf{y})}f$ at point $(\bm{\tilde{\psi}},\mathbf{\tilde{y}},\mathbf{\tilde{H}} )$ is invertible.
	
	By some calculations, we obtain
	\begin{equation}\label{eq:deriv}
	D_{(\bm{\psi},\mathbf{y})}f=\begin{bmatrix} \frac{\partial f}{\partial \bm{\psi}} & \frac{\partial f}{\partial \mathbf{y}} \end{bmatrix}=\begin{bmatrix} * & \mathbf{I}_M\\\mathbf{G}&\mathbf{0}
	\end{bmatrix}\in \mathbb{C}^{2M\times 2M},
	\end{equation}
	where%
	\begin{equation}\label{eq:G2}
	\begin{aligned}
	\mathbf{G}=&
	~\mathbf{C}_{\bm{\psi}}\mathbf{\text{Re}\{ A_H \}}\mathbf{C}_{\bm{\psi}}- \mathbf{S}_{\bm{\psi}}\textrm{Diag}(\mathbf{\text{Re}\{ A_H \}}\mathbf{s}_{\bm{\psi}})\\&
	+ \mathbf{S}_{\bm{\psi}}\mathbf{\text{Im}\{ A_H \}}\mathbf{C}_{\bm{\psi}}+ \mathbf{C}_{\bm{\psi}}\textrm{Diag}(\mathbf{\text{Im}\{ A_H \}}\mathbf{s}_{\bm{\psi}})\\&
	- \mathbf{C}_{\bm{\psi}}\mathbf{\text{Im}\{ A_H \}}\mathbf{S}_{\bm{\psi}}- \mathbf{S}_{\bm{\psi}}\textrm{Diag}(\mathbf{\text{Im}\{ A_H \}}\mathbf{c}_{\bm{\psi}})\\&
	+ \mathbf{S}_{\bm{\psi}}\mathbf{\text{Re}\{ A_H \}}\mathbf{S}_{\bm{\psi}}- \mathbf{C}_{\bm{\psi}}\textrm{Diag}(\mathbf{\text{Re}\{ A_H \}}\mathbf{c}_{\bm{\psi}})\\&
	+ \mathbf{S}_{\bm{\psi}}\textrm{Diag}(\mathbf{\text{Im}\{ b_H \}})+\mathbf{C}_{\bm{\psi}}\textrm{Diag}(\mathbf{\text{Re}\{ b_H \}}).
	\end{aligned}
	\end{equation}
	By \eqref{eq:define_zeropoint}, we have
	\begin{equation}\label{eq:tilde_eq_real}
	\begin{aligned}
	\mathbf{\text{Re}\{ A_{\tilde{H}}\}}\mathbf{{c}}_{\bm{\tilde{\psi}}}-\mathbf{\text{Im}\{ A_{\tilde{H}}\}}\mathbf{{s}}_{\bm{\tilde{\psi}}}=\mathbf{\text{Re}\{ A_{\tilde{H}}\}},\\
	\mathbf{\text{Im}\{ A_{\tilde{H}}\}}\mathbf{{c}}_{\bm{\tilde{\psi}}}+\mathbf{\text{Re}\{ A_{\tilde{H}}\}}\mathbf{{s}}_{\bm{\tilde{\psi}}}=\mathbf{\text{Im}\{ b_{\tilde{H}}\}}.
	\end{aligned}
	\end{equation}
	Substituting \eqref{eq:tilde_eq_real} into \eqref{eq:G2}, we obtain 
	\begin{align*}
	\mathbf{G}|_{(\bm{\tilde{\psi}},\mathbf{\tilde{y}},\mathbf{\tilde{H}})}=\begin{bmatrix}\mathbf{C}_{\bm{\tilde{\psi}}}&\mathbf{S}_{\bm{\tilde{\psi}}}\end{bmatrix}
	\begin{bmatrix}\mathbf{\text{Re}\{ A_{\tilde{H}}\}}&-\mathbf{\text{Im}\{ A_{\tilde{H}}\}}\\\mathbf{\text{Im}\{ A_{\tilde{H}}\}}&\mathbf{\text{Re}\{ A_{\tilde{H}}\}}\end{bmatrix}
	\begin{bmatrix}\mathbf{C}_{\bm{\tilde{\psi}}}\\\mathbf{S}_{\bm{\tilde{\psi}}}\end{bmatrix}.
	\end{align*}
	Since $\mathbf{A_{\tilde{H}}}=\mathbf{\tilde{H}^\dagger P \tilde{H}}$ is positive definite and the matrix $\left[\mathbf{C}_{\bm{\tilde{\psi}}}, \mathbf{S}_{\bm{\tilde{\psi}}}\right]$ is of full row rank, we immediately have
	$$\mathbf{x}^T\mathbf{G}|_{(\bm{\tilde{\psi}},\mathbf{\tilde{y}},\mathbf{\tilde{H}})}\mathbf{x}>0,~\forall~\mathbf{x}\in\mathbb{R}^{M},$$
	which further implies that $\mathbf{G}|_{(\bm{\tilde{\psi}},\mathbf{\tilde{y}},\mathbf{\tilde{H}})}$ is invertible. Consequently, the Jacobian matrix $D_{(\bm{\psi},\mathbf{y})}f$ in \eqref{eq:deriv} at point $(\bm{\tilde{\psi}},\mathbf{\tilde{y}},\mathbf{\tilde{H}})$ is invertible. The proof of Claim 1 is completed.
\end{proof}
%

{Let $\mathbf{F}$ and $\mathbf{D}$ be the Jacobian matrices of function $f$ (defined in \eqref{eq:definef}) with respect to $\textrm{vec}\left(\mathbf{H}\right)$ and $(\bm{\psi},\mathbf{y})$ at point $(\bm{\tilde{\psi}},\mathbf{\tilde{y}},\mathbf{\tilde{H}})$, respectively.} From the proof of Claim 1, we know that both $\mathbf{F}$ and $\mathbf{D}$ are well defined.

\noindent \textbf{Claim 2.} Consider $\mathbf{H}\in{\mathcal{H}}\subseteq \mathbb{C}^{(K-1)\times M},$ where ${\mathcal{H}}$ is the one in Claim 1. 
	Suppose \begin{equation}\label{conditionpositive}\lambda_{\textrm{min}}(\mathbf{H}^\dagger\mathbf{PH})>\min_{1\leq m\leq M}\left\{-\text{Re}\left\{\left[ \mathbf{D}^{-1}\mathbf{F}\textrm{vec}\left(\mathbf{H} - \mathbf{\tilde H}\right) \right]_{M+m} \right\} \right\}.\end{equation} 
	Then, $\mathbf{W}$ and $\mathbf{y}$ defined in Claim 1 satisfy \eqref{strictpositive}.
\begin{proof}
	It follows from Claim 1 that   
	$$\mathbf{y}_{m}{{(\mathbf{H})}}:=\left[d_{\mathbf{y}}(\mathbf{H})\right]_m,~m=1,2,\ldots,M$$ are continuously differentiable functions with respect to $\mathbf{H}\in{\mathcal{H}}\subseteq\mathbb{C}^{(K-1)\times M}.$ 
	Therefore,~for $m=1,2,\ldots,M,$ we get
	\begin{align}
	\mathbf{y}_m{{(\mathbf{H})}}
	& =  \mathbf{y}_m(\mathbf{\tilde H}) + \text{Re}\left\{ \textrm{Tr}\left(\nabla \mathbf{y}_m(\mathbf{\tilde H}) {{ \left(\mathbf{H}-\mathbf{\tilde H}\right)}}\right) \right\} + {{o}}(\|{ {\mathbf{H}-\mathbf{\tilde H}}}\|)\nonumber \\
	& = -\text{Re}\left\{\left[ \mathbf{D}^{-1}\mathbf{F}\textrm{vec} {{\left(\mathbf{H}-\mathbf{\tilde H}\right)}} \right]_{M+m} \right\}+ {{o}}(\|{{\mathbf{H}-\mathbf{\tilde H}}}\|),\label{taylory}
	\end{align}
	where \eqref{taylory} is due to the {differentiation rule of the implicit function} and $\mathbf{y}_m(\mathbf{\tilde H})=0.$ 
	Combining \eqref{taylory} and \eqref{conditionpositive}, we immediately have
	\begin{align}
	\mathbf{W}^\dagger\mathbf{H}^\dagger\mathbf{PHW}+\textrm{Diag}(\mathbf{y}_{1:M})
	=\mathbf{W}^\dagger\left[\mathbf{H}^\dagger\mathbf{PH}+\textrm{Diag}(\mathbf{y}_{1:M})\right]\mathbf{W}\succ \mathbf{0},
	\end{align}which shows that $\mathbf{W}$ and $\mathbf{y}$ satisfy \eqref{strictpositive}. The proof of Claim 2 is completed.
\end{proof}

\begin{acknowledgements}
We would like to thank Professor Stephen Boyd of Stanford University for several useful discussions.
\end{acknowledgements}

\bibliography{refs0429}

\bibliographystyle{spmpsci}


\end{document}